\documentclass[a4, 11 pt]{article}

\usepackage{fullpage}
\usepackage{parskip}
\usepackage{amsmath}
\usepackage{mathtools}
\usepackage{amssymb}
\usepackage{todonotes}
\usepackage{tabularx}
\usepackage[toc,page]{appendix}
\usepackage{xspace}

\usepackage{bbm}
\usepackage{amsthm}
\usepackage{enumerate}
\usepackage[]{hyperref}

\usepackage{comment}
\usepackage[classfont=sanserif,langfont=sanserif,funcfont=sanserif]{complexity}
\usepackage{thmtools, thm-restate}
\usepackage{xcolor}

\newtheorem{definition}{Definition}
\newtheorem*{theorem*}{Theorem}
\newtheorem{theorem}{Theorem}

\newtheorem*{conjecture*}{Conjecture}
\newtheorem{corollary}{Corollary}
\declaretheorem[name=Theorem]{thm}
\declaretheorem[name=Lemma]{lem}
\declaretheorem[name=Conjecture]{conj}

\newcommand{\APSP}{\textsc{APSP}}
\newcommand{\threeSUM}{\textsc{3SUM}}
\newcommand{\QthreeSUMconjecture}{Quantum-3SUM-Conjecture}
\newcommand{\CthreeSUMconjecture}{Classical-3SUM-Conjecture}
\newcommand{\sortedThreeSUM}{\textsc{Sorted-3SUM}}
\newcommand{\sortedThreeSUMThreeList}{\textsc{Sorted-3SUM-3ListVersion}}
\newcommand{\orderedThreeSUM}{\textsc{Ordered-3SUM}}
\newcommand{\uniqueSortedSUM}{\textsc{Unique-Sorted-3SUM}}
\newcommand{\threeSUMpatrascuVersion}{\textsc{3SUM'}}
\newcommand{\threeSUMthreeList}{\textsc{3SUM-3ListVersion}}
\newcommand{\ConvolutionThreeSum}{\textsc{Convolution-3SUM}}
\newcommand{\ZeroWtTriangle}{\textsc{0-Edge-Weight-Triangle}}
\newcommand{\geomBase}{\textsc{GeomBase}}
\newcommand{\sortedGeomBase}{\textsc{Sorted-GeomBase}}
\newcommand{\separator}{\textsc{Separator}}
\newcommand{\stripsCoverBox}{\textsc{Strips-Cover-Box}}
\newcommand{\trianglesCoverTriangle}{\textsc{Triangles-Cover-Triangle}}
\newcommand{\threePointsOnLine}{\textsc{3-Points-on-Line}}
\newcommand{\PointOnThreeLines}{\textsc{Point-on-3-Lines}}
\newcommand{\GeneralCovering}{\textsc{General-Covering}}
\newcommand{\PointCovering}{\textsc{Point-Covering}}
\newcommand{\VisibilityBetweenSegments}{\textsc{Visibility-Between-Segments}}

\newcommand{\HoleInUnion}{\textsc{Hole-In-Union}}
\newcommand{\TriangleMeasure}{\textsc{Triangle-Measure}}
\newcommand{\VisibilityFromInfinity}{\textsc{Visibility-From-Infinity}}
\newcommand{\VisibleTriangle}{\textsc{Visible-Triangle}}
\newcommand{\PlanarMotionPlanning}{\textsc{Planar-Motion-Planning}}
\newcommand{\threeDmotionPlanning}{\textsc{3D-Motion-Planning}\xspace}

\newcommand{\ket}[1]{|{#1}\rangle}

\newcommand{\DSone}{\textsc{DS-1}}
\newcommand{\DStwo}{\textsc{DS-2}}
\newcommand{\traceNorm}[2]{||#1-#2||}

\newcommand{\subha}[1]{\textcolor{red}{Subha: #1}}
\newcommand{\florian}[1]{\textcolor{blue}{Florian: #1}}

\newcommand\backSlash{\char`\\}
\newcommand{\walkUnitary}{U_{\textit{walk}}}
\newcommand{\phaseFlipUnitary}{U_{\textit{phaseFlip}}}
\newcommand{\someUnitary}{U}
\newcommand{\inputsetS}{S}
\newcommand{\setupTime}{T_\textit{setup}}

\newcommand{\UnitaryTime}{T_\textit{unitary}}
\newcommand{\countChildren}{\texttt{COUNT\_CHILDREN}}
\newcommand{\countUniqueChildren}{\texttt{COUNT\_UNIQUE\_CHILDREN}}

\newcommand{\sortedThreeSumTime}{T_\textit{Sorted-3SUM}}
\newcommand{\ConvolutionThreeSumTime}{T_\textit{c-3SUM}}
\newcommand{\reducedTo}[1]{\ensuremath{\lll_{#1}}}
\newcommand{\equivalentTo}[1]{\ensuremath{==_{#1}}}
\newcommand{\PRone}{\textsc{PR1}}
\newcommand{\PRtwo}{\textsc{PR2}}

\newcommand{\remainder}{m}
\newcommand{\quotient}{q}

\newcommand{\nodeSkipList}[1]{\ensuremath{\textit{node}_{#1}}}
\newcommand{\nodeNeighbourLeft}[2]{\ensuremath{\textit{node}_{#1}\xrightarrow[]{#2}\textit{left}}}
\newcommand{\nodeNeighbourRight}[2]{\ensuremath{\textit{node}_{#1}\xrightarrow[]{#2}\textit{right}}}

\newcommand{\variableForIndexing}[2]{v_{#2}^{#1}}

\usepackage{mathtools}
\DeclarePairedDelimiter\ceil{\lceil}{\rceil}
\DeclarePairedDelimiter\floor{\lfloor}{\rfloor}

\newcommand{\eps}{\varepsilon}
\newcommand{\ZO}{\{0,1\}}
\newcommand{\tO}{\widetilde O}

\newcommand{\Patrascu}{Pătraşcu\xspace}
\newcommand{\pth}{^{\text{th}}}
\usepackage{mathabx,makecell}
\begin{document}
\pagenumbering{roman}

\title{Limits of quantum speed-ups\\ for computational geometry and other problems:\\Fine-grained complexity via quantum walks}

\author{Harry Buhrman${^*}$ \\ \href{mailto:harry.buhrman@cwi.nl}{harry.buhrman@cwi.nl} 
 \and Bruno Loff${^\delta}$ \\ \href{mailto:bruno.loff@gmail.com}{bruno.loff@gmail.com}
 \and Subhasree Patro${^*}$ \\ \href{mailto:subhasree.patro@cwi.nl}{subhasree.patro@cwi.nl}
   \and Florian Speelman$^*$ \\ \href{mailto:f.speelman@uva.nl}{f.speelman@uva.nl} \\[0.35cm]
   $^*$QuSoft, University of Amsterdam, 
CWI Amsterdam, $^\delta$University of Porto}

\date{}
\maketitle

\begin{abstract}
    Many computational problems are subject to a quantum speed-up: one might find that a problem having an $O(n^3)$-time or $O(n^2)$-time classic algorithm can be solved by a known $O(n^{1.5})$-time or $O(n)$-time quantum algorithm. The question naturally arises: \textit{how much quantum speed-up is possible?}
    
    The area of fine-grained complexity allows us to prove optimal lower-bounds on the complexity of various computational problems, based on the conjectured hardness of certain natural, well-studied problems. This theory has recently been extended to the quantum setting, in two independent papers by Buhrman, Patro and Speelman \cite{BuhrmanPatroSpeelman}, and by Aaronson, Chia, Lin, Wang, and Zhang \cite{Aaronson-ClosestPair-2019}.
    
    In this paper, we further extend the theory of fine-grained complexity to the quantum setting. A fundamental conjecture in the classical setting states that the \threeSUM{} problem cannot be solved by (classical) algorithms in time $O(n^{2-\eps})$, for any $\eps > 0$. We formulate an analogous conjecture, the \QthreeSUMconjecture{}, which states that there exist no sublinear $O(n^{1-\eps})$-time quantum algorithms for the \threeSUM{} problem. 
    
    Based on the \QthreeSUMconjecture{}, we show new lower-bounds on the time complexity of quantum algorithms for several computational problems. Most of our lower-bounds are optimal, in that they match known upper-bounds, and hence they imply tight limits on the quantum speedup that is possible for these problems.
    
    These results are proven by adapting to the quantum setting known classical fine-grained reductions from the \threeSUM{} problem. This adaptation is not trivial, however, since the original classical reductions require pre-processing the input in various ways, e.g.\ by sorting it according to some order, and this pre-processing (provably) cannot be done in sublinear quantum time.
    
    We overcome this bottleneck by combining a quantum walk with a classical dynamic data-structure having a certain ``history-independence'' property. This type of construction has been used in the past to prove upper bounds, and here we use it for the first time as part of a reduction. This general proof strategy allows us to prove tight lower bounds on several computational-geometry problems, on \ConvolutionThreeSum{} and on the \ZeroWtTriangle{} problem, conditional on the \QthreeSUMconjecture{}.

    
    We believe this proof strategy will be useful in proving tight (conditional) lower-bounds, and limits on quantum speed-ups, for many other problems.
\end{abstract}











\tableofcontents{}

\newpage
\pagenumbering{arabic}
\setcounter{page}{1}

\section{Introduction}
\label{sec:Introduction}

The world is investing in quantum computing because of so-called \textit{quantum speed-ups}: quantum algorithms can solve many computational problems faster than their classical counterparts. However, the amount of speed-up that is possible varies among different computational problems. It is expected that quantum computers will remain an expensive resource for decades to come, and the extent to which a quantum speed-up is possible, or not possible, may one day be a key factor in deciding whether or not to invest in the use of a quantum computation in for example an industrial setting. It is therefore essential to understand how much quantum speed-up is possible for a specific computational problem, and for this purpose we need to have tight upper and lower-bounds on both classical and quantum algorithms.

Sadly, the state of affairs is such that we do not even know how to prove super-linear time lower-bounds (e.g., on a classical random-access machine). Hence, there are some computational problems which \textit{do} have polynomial-time (e.g.\ quadratic-time) algorithms, classical or quantum, and these algorithms are conjectured to be optimal, but we presently have no way of proving this.

The theory of fine-grained complexity has been developed in the last decade to overcome this problem. Analogous to how \NP{}-completeness allows us to prove super-polynomial lower-bounds, fine-grained complexity allows us to prove tight fixed-polynomial (e.g. quadratic) lower-bounds on the time complexity of many problems in \P{}, conditioned on hardness conjectures for a few natural, well-studied problems. Three central hardness conjectures are the strong exponential-time hypothesis (SETH) for satisfiability, a conjectured cubic-hardness for the all-pairs shortest-path problem (\APSP{}) and a conjectured quadratic hardness for the \threeSUM{} problem.\footnote{The survey by Vassilevska Williams contains an overview of many results within this area \cite{Survey-VVWilliams-15}.}

Recently, two independent works initiated the study of fine-grained complexity in the quantum setting. Both works studied quantum variants of SETH, and used these variants to prove (often tight) bounds on how much quantum speed-up is possible for various problems. Aaronson, Chia, Lin, Wang, and Zhang~\cite{Aaronson-ClosestPair-2019} presented linear quantum time lower bounds for Closest Pair, Bichromatic Closest Pair, and Orthogonal Vectors, conditioned on the quantum hardness of the Satisfiability (SAT) problem. In the same paper, they also present matching quantum upper bounds for these problems. 
Simultaneously, Buhrman, Patro, and Speelman~\cite{BuhrmanPatroSpeelman} presented a framework for proving quantum time lower bounds for many problems in $\P$ conditioned on quantum hardness of variants of SAT, which they used to prove an $n^{1.5}$ quantum time lower bound for the Edit Distance and the Longest Common Subsequence problems.



In this work, we explore quantum fine-grained reductions to derive quantum time lower-bounds for several problems in \P{}, conditioned on the conjectured quantum hardness for the \threeSUM{} problem. These lower-bounds often tightly match upper-bounds given by known quantum algorithms, and similar tight upper and lower-bounds have also been proven in the classical setting. Together, these tight classical and quantum bounds are finally able to tell us exactly how much quantum speed-up is possible for various problems, which is the main goal of this line of research.

\subsection{The conjectured hardness of \threeSUM{}}

The \threeSUM{} problem is defined as follows: We are given as input a list $S$ of $n$ integers, which we may assume to be between $- n^3$ and $n^3$,\footnote{This is because the \threeSUM{} problem over lists with larger integers can be reduced to the \threeSUM{} problem on $n^3$-bounded integers by a simple hashing technique.} and we wish to know if there exist $a, b, c$ in $S$ such that $a+b+c=0$. There is a simple classical algorithm that solves this problem in $\tilde O(n^2)$ time, but even after many years of interest in the problem, the exponent has not been reduced. The conjecture naturally arises that there is no $\epsilon>0$, such that \threeSUM{} can be solved in $O(n^{2-\epsilon})$ classical time. We refer to this conjecture as the \CthreeSUMconjecture{}. Using this conjecture, one can derive conditional classical lower bounds for a vast collection of computational geometry problems, dynamic problems, sequence problems, etc.\ \cite{Overmars-ComputationalGeometry-1995, Williams-FindingTypesOfTriangles-2009, Patrascu-Convoluted3SUM-2010, Survey-VVWilliams-15}. 

However, the \CthreeSUMconjecture{} no longer holds true in the quantum setting, as there is a faster \textit{quantum} algorithm for \threeSUM{}: we may use Grover search as a subroutine in the $\tilde O(n^2)$ classical algorithm to solve the problem in $\tilde{O}(n)$ quantum time. Apart from this quadratic speedup, no further improvement to the quantum-time upper-bound is  known.
It is worth mentioning that there is a sub-linear $O(n^{3/4})$ quantum {\em query} algorithm for computing \threeSUM{} \cite{Ambainis-ElementDistinctness-2004, Andrew-SubsetFinding-2005} --- with a matching lower bound of $\Omega(n^{3/4})$ \cite{Belovs-KSUMLowerBound-2012} --- this query algorithm however, is not time efficient.
Consequently, it was conjectured \cite{Ambainis-QuantumGeometryProblems-2020} that the \threeSUM{} problem cannot be solved in sub-linear quantum time in the QRAM model:
\begin{conj}[\QthreeSUMconjecture{} \cite{Ambainis-QuantumGeometryProblems-2020}]
\label{conj:3SUM}
There does not exist a $\delta>0$ such that \threeSUM{} on a list of $n$ integers can be solved in $O(n^{1-\delta})$ quantum time in the QRAM model.
\end{conj}

It is then natural to try to extend the classical \threeSUM{}-based lower bounds to the quantum setting, and one may at first expect this task to be a simple exercise. However, one soon realizes that none of the existing classical reductions can be easily adapted to the quantum regime. Indeed, most of the existing classical reductions begin by pre-processing the input in some way, e.g., by sorting it according to some ordering. This is not an issue in the classical setting, as the classical conjectured lower bound for \threeSUM{} is quadratic. Hence, the classical reductions can accommodate any pre-processing of the input that takes sub-quadratic time, such as e.g. sorting. However, this pre-processing becomes problematic in the quantum setting, since here we will need a sublinear-time quantum reduction, and even simple sorting requires linear quantum time on a quantum computer. \cite{Hoyer-QuantumSorting-2001}. 

We present a workaround for this problem.  The idea of the proof is to adapt Ambainis' quantum walk algorithm for element distinctness~\cite{Ambainis-ElementDistinctness-2004}. To overcome for example, the case of first sorting the input, instead of having the reduction sort the entire list, we combine a data structure for dynamic sorting together with a quantum walk algorithm. As we will show, this approach only needs the reduction to sort a only small part of the input and thus allows us to show that \threeSUM{} remains hard, even when the entire input is sorted.  As we will see, this idea can be extended to allow for any ``structuring'' of the input (not just sorting) which can be implemented by a dynamic data structure obeying a certain ``history-independence'' property. The proof will be sketched in Section~\ref{sec:OrderingDoesntHelpSUM}.

This \textit{quantum-walk plus data-structure} proof strategy has been used to prove upper-bounds on other problems (e.g., for the closest-pair problem \cite{Aaronson-ClosestPair-2019}), and here we use it for the first time as part of a reduction in order to obtain a lower-bound. We expect that the same strategy will be applicable to other quantum fine-grained reductions, and our hope is that this will give rise to a landscape of results, that establish (conditional) tight lower-bounds for quantum algorithms. This will in turn  precisely answer the question of how much quantum speed-up is possible for a variety of computational problems.

Using this strategy we are able to show that various ``structured'' versions of \threeSUM{} are as hard as the original (unstructured) \threeSUM{} problem, even in the quantum case. Once we have shown that these structured versions of \threeSUM{} are hard, we may then construct direct quantum adaptations of the classical reductions, to show the quantum hardness of several computational-geometry problems (Section~\ref{sec:ImplicationsOfSortedUniqueSUM}), of \ConvolutionThreeSum{} and of the \ZeroWtTriangle{} problem (Section~\ref{sec:TimeLowerBoundsThreeSum}). This enables us to prove quantum time lower-bounds for these problems, conditioned on the \QthreeSUMconjecture{}. We give further details in Section~\ref{sec:SummaryQuantumSumHardProblems}.

\subsection{Main Idea: Reductions via Quantum Walks}
\label{sec:OrderingDoesntHelpSUM}

The \CthreeSUMconjecture{} states that there is no sub-quadratic classical algorithm to solve the \threeSUM{} problem. However, the statement of this conjecture can be shown to be equivalent to the same statement for a promise version of \threeSUM{} where the input $S$ is sorted. That is because, if there was a sub-quadratic algorithm for \threeSUM{} on sorted inputs, then given any (unsorted) input one can first sort the entire input with additional $O(n \log n)$ pre-processing time, and then use the sub-quadratic algorithm for sorted \threeSUM{}, resulting in a sub-quadratic algorithm for unsorted \threeSUM{}. In fact, one can make a more general statement in this regard. An input to the \threeSUM{} problem is a list $S \in \{-n^3,\dots, n^3\}^n$ of $n$ integers (possibly with repetitions). One may consider a family $\{q_{i}\}$ of \textit{queries}, i.e., each $q_{i}: \{-n^3, \dots, n^3\}^n\to A_{i}$ is a function on lists of integers, for some set $A_{i}$ of possible answers to the query. (For example, $q_{i}(S)$ could output the $i$-th smallest integer in $S$.) We may then ask about static data-structures that allow us to efficiently answer these queries. (For example, we may consider the sorted version of $S$ to be a data-structure that allows us to efficiently obtain the $i$-th smallest element of $S$.) Then we may generally state that, if it is possible to preprocess an input $S$ in sub-quadratic time to produce a static data-structure that allows us to answer any query in $n^{o(1)}$ time, then the ``structured'' variant of \CthreeSUMconjecture{}, where we give the algorithm access to all the queries $q_{i}(S)$ for free, is equivalent to the original version of \CthreeSUMconjecture{}.

Most of the known fine-grained reductions from \threeSUM{}, in the classical setting, can be explained in the following way: one first shows that a certain ``structured'' variant of \threeSUM{} is just as hard as the original \threeSUM{} problem, and then one reduces the structured variant of \threeSUM{} to another problem. While for some reductions \cite{Overmars-ComputationalGeometry-1995} require the input list to be sorted in the usual order of the integers, other reductions require the input to be structured in some other way, for example, reductions in \cite{Patrascu-Convoluted3SUM-2010, Williams-FindingTypesOfTriangles-2009} require that the elements are hashed into buckets and every element in the bucket can be accessed efficiently.

The reduction from ``unstructured'' to ``structured'' \threeSUM{} is usually trivial to do in classical sub-quadratic time, but not so in quantum sub-linear time (e.g., a quantum computer cannot sort in sublinear time \cite{Hoyer-QuantumSorting-2001}). This is the main difficulty in translating the classical reductions to the quantum setting.

Our main observation is that, if a certain analogous \textit{dynamic} data-structure problem can be solved efficiently by a dynamic data-structure possessing a certain ``history-independence'' property, then it is possible to use a quantum walk in order to show that the ``structured'' variant of \QthreeSUMconjecture{}, where we give the algorithm access to the queries for free, is equivalent to the original unstructured version of the \QthreeSUMconjecture{}. It is this insight that underlies all of our reductions, and which we expect will open up the  way to many other fine-grained reductions in the quantum setting.

One might informally state our observation as follows.

\begin{restatable}[informal]{thm}{HardnessOfOrderedSUM}
\label{thm:HardnessOfOrderedSUM}
Let $\{q_{i}\}$ be a collection of queries over \threeSUM{} inputs, i.e., each $q_{i}$ is a function over inputs $S \in\{-n^3,\ldots,n^3\}^n$ for \threeSUM{}. Suppose that there exists an efficient classical dynamic data-structure that allows us to answer the queries $q_{i}$, under updates to $S$, where an update consists of replacing an element in the list $S$ by a different element. By efficient we mean that any query or update can be carried out in $n^{o(1)}$ time. Suppose further that the dynamic data structure satisfies the ``history-independence'' property\footnote{Also mentioned in \cite{Ambainis-ElementDistinctness-2004,Aaronson-ClosestPair-2019}}, which means that the data structure corresponding to each set $S$ has a unique representation in memory, which only depends  on the current value of $S$ (so it is independent of the initial value of $S$, and of the subsequent updates which resulted in the current value of $S$).

Then, conditioned on the \QthreeSUMconjecture{}, \threeSUM{} cannot be done in $O(n^{1-\eps})$ quantum time, for any $\eps > 0$, even if the queries $q_{i}(S)$ can be done at unit cost.
\end{restatable}

Hereafter, we refer to these versions of \threeSUM{}, where queries $q_i(S)$ have unit cost, as ``structured'' versions of the \threeSUM{} problem. To be clear, by \textit{being able to do the queries at unit cost}, we mean that the algorithm is given access to an oracle gate, implementing the unitary transformation:
\[
\ket{i, b} \mapsto \ket{i, b \oplus q_i(S)}.
\]

The distinction between an arbitrary dynamic data-structure and a history-independent solution should be understood as follows. Generally speaking, a solution to a dynamic data-structure problem could represent data in a way which depends on the specific sequence of updates which were applied to the initial data. For example, self-balancing trees are a solution to the dynamic sorting problem, but the specific balancing of the tree which is kept in memory depends on the sequence of updates which were applied, so different sequences of insertions and deletions might lead to the same list, but will nonetheless be represented differently in memory. A history-independent data-structure, however, has fixed a-priori representations for each possible data value. So, for example, in the dynamic sorting problem, a history-independent data-structure must represent each possible list in a unique, or canonical way in memory.

\paragraph{Our idea.} Let $S = (x_1, \ldots, x_n)$ be an unstructured input to \threeSUM{}. We will now discuss quantum query algorithms for solving \threeSUM{}. Such algorithms can access the input only via a unitary $\ket{i, b} \mapsto \ket{i, b \oplus x_i}$. Each application of this unitary is called a \textit{query}. But, in accordance to data-structure nomenclature, we have also called \textit{queries} to the functions $q_i$. So to distinguish the two, in this section we will use \textit{input queries} to refer to queries to the input, in the sense of query complexity, and let us use \textit{data-structure queries}, to refer to the values $q_i(S)$.

Consider the quantum walk algorithm for Element Distinctness by Ambainis~\cite{Ambainis-ElementDistinctness-2004}. It was observed by Childs and Eisenberg~\cite{Andrew-SubsetFinding-2005} that this algorithm can be used to solve any problem, such as \threeSUM{}, where we wish to find a constant-size subset that satisfies a given property.  Although this algorithm is optimal and sub-linear for \threeSUM{} when we only measure the number of input queries (it uses $\Theta(n^{3/4})$ input queries, and this is required \cite{Belovs-KSUMLowerBound-2012}), the algorithm still requires linear time, essentially because an $\Omega(n^{1/4})$-time operation is performed between each input query.

 This optimal query algorithm for \threeSUM{} is a quantum walk on the \textit{Johnson graph}, namely, the graph of $n \choose r$ vertices with each vertex of the graph labelled by an $r$-sized subset of $[n]$, and where there is an edge between two vertices if and only if the two corresponding sets differ by exactly two elements. This resulting graph $J(n,r)$ is a good-enough expander, so that a quantum walk will be able to find an $r$-sized subset of $[n]$ containing indices to three elements of $S$ that sum to zero, in queries sublinear in $n$.\footnote{For an excellent introduction to quantum walks, see Chapter 8 of Ronald de Wolf's lecture notes \cite{deWolf-QuantumLectureNotes-2021}.} To do so, the quantum-walk algorithm maintains the list of values $(x_{i_1}, \ldots, x_{i_r})$ entangled together with the basis state representing the current $r$-sized subset $\{i_1, \ldots, i_r\} \subseteq [n]$ that is being traversed. Using this list of values, as a part of the quantum walk algorithm, a subroutine checks (in superposition) if there is a \threeSUM{} solution in $(x_{i_1}, \ldots, x_{i_r})$. While this step requires no additional input queries, so the total number of input queries is $O(n^{3/4})$, the actual implementation of this subroutine requires a significant amount of time (namely time $r = \Omega(n^{1/4})$), which then makes the resulting quantum walk algorithm for \threeSUM{} linear, at best.

It is this subroutine, i.e.~the subroutine that checks for a \threeSUM{} solution in the $r$-sized set of values, that we would like to further speed up. Now suppose that we had a faster-than-linear algorithm for a ``structured'' version of \threeSUM{}. I.e., the algorithm works in sublinear time, provided it is given certain data-structure queries $q_i(S)$ as part of the input. Now, if we could efficiently answer these data-structure queries at any point during the entire quantum walk, then we could use this faster-than-linear algorithm to speed-up the subroutine. To do so, we need a dynamic data structure that allows us to efficiently answer the data-structure queries, under the kind of updates that are required at each step of the quantum walk. For the quantum walk on the Johnson graph, each update corresponds to replacing a single element in the list of values $(x_{i_1}, \ldots, x_{i_r})$.

An important detail remains: in order for the quantum walk to work, it is necessary that there is a unique basis state corresponding to each node in the quantum-walk graph (otherwise we won't have the desired amplitude interference). It is for this reason that the dynamic data-structure structure is required to have a history-independence property.



\begin{proof}[Proof of Theorem \ref{thm:HardnessOfOrderedSUM} (sketch)]
In order to prove this theorem, we will first go through the steps of the more general version of Ambainis' quantum walk algorithm for Element Distinctness given by \cite{Andrew-SubsetFinding-2005}.  

Let $\inputsetS \in \{-n^3, \ldots, n^3\}^n$ be an input to the \threeSUM{} problem. Let $r=n^\beta$ for some $\beta \in (0,1)$ which will be fixed later (so that $r$ is an integer). The graph $G$ used in Ambainis' construction is a Johnson graph $J(n,r)$ with vertices all labelled by $r$-sized subsets of $[n]$. Let $V, V' \subset [n]$ with $|V|=|V'|=r$. Vertices labelled by $V$ and $V'$ are connected if and only if $|V \cap V'|=r-1$, i.e.,~$V'$ can be obtained by replacing a single element of $V$.

Given a subset $I \subset [n]$, we use $\inputsetS[I]$ to denote all the elements $\inputsetS[i], i \in I$. Now suppose we have a history-independent classical dynamic data structure for answering a family of data-structure queries $\{q_i\}$, where each $q_i:\{-n^3, \ldots, n^3\}^r$. For $V \subseteq [n]$ of size $|V| = r$, let $D(S[V])$ denote the (unique) state of the data-structure corresponding to $S[V]$. I.e., given $D(S[V])$, we are able to answer any query $q_i(S[V])$ in time $n^{o(1)}$. And if we change $V$ to $V'$ by replacing a single element of $V$, we are able to update $D(S[V])$ to $D(S[V'])$, also in time $n^{o(1)}$.

To define a quantum walk on $G$, define an orthonormal basis of quantum states $\ket{V}$, one for each $r$-subset $V$. The key idea is to store values from the list, and the contents of the data-structure, along with the subset $V$.
So the full quantum state has the form $\ket{V, \inputsetS[V], D(\inputsetS[V]), k}$ where $k \in [n]$.
If $|V|=r$ then $k$ denotes an element in $[n]\backSlash{}V$ to be added to $V$. We say a vertex $V$ is marked if $\inputsetS[V]$ is a positive \threeSUM{} instance (of smaller size), i.e., if there are $p,q,r \in V$ such that $\inputsetS[p]+\inputsetS[q]+\inputsetS[r]=0$.

The quantum walk algorithm is analogous to Grover's algorithm, where the aim is to make the amplitude on marked vertices large enough that with very high probability\footnote{Throughout the paper we say that something holds “with high probability” if it holds with probability at least $1-o(1)$.} the final measurement collapses on a marked vertex, i.e., a vertex labelled by an $r$-subset that contains a solution to \threeSUM{} problem. The algorithm starts with a state
\begin{equation}
  \ket{s}=\frac{1}{\sqrt{c}}\sum_{|V|=r} \ket{V, \inputsetS[V], D(\inputsetS[V])} \sum_{k \notin V}\ket{k},  
\end{equation}
which is a uniform superposition of all the states on subsets of size $r$ and $c=(n-r){n \choose r}$ is the normalization constant. 

There are two main operations in this algorithm: A walk operation $\walkUnitary$ and a phase flip operation $\phaseFlipUnitary$ which is
\begin{equation}
\label{eq:LooksOfPhaseFlipUnitary}
  \phaseFlipUnitary\ket{V,S[V], D(\inputsetS[V])}=  
  \begin{cases}
      -\ket{V,S[V], D(\inputsetS[V])} & \text{if }V\text{ is marked}\\
      \ket{V,S[V], D(\inputsetS[V])} & \text{if }V\text{ is not marked}.
    \end{cases}
\end{equation}
The full algorithm is $(\walkUnitary^{t_1} \phaseFlipUnitary)^{t_2}$ where $t_1=O(\sqrt{r})$ and $t_2=O((n/r)^{1.5})$. The total time taken by the algorithm is 
\begin{equation}
    \setupTime(\ket{s}) + t_1\cdot t_2 \cdot \UnitaryTime(\walkUnitary) + t_2\cdot \UnitaryTime(\phaseFlipUnitary),
\end{equation}
where $\setupTime(\ket{s})$ denotes the time taken to setup the initial state $\ket{s}$ that also includes the time taken to query values of the subset of indices of size $r$. The term $\UnitaryTime(\someUnitary)$ denotes the number of elementary gates required to implement a unitary $\someUnitary$. 

In the setup phase, for every vertex $V$ we initialize the dynamic data-structure corresponding to $S[V]$. We may think of $S[V]$ as obtained via the $(0, \ldots, 0)$ list by updating each position $i$ with $S[i]$.
Hence, the setup time for each vertex, which consists of computing $D(S[V])$ for all $V$ in superposition, is at most $r n^{o(1)}$.

Now, because the data structure supports efficient updates, the $\walkUnitary$ unitary can be implemented in time $n^{o(1)}$. It in is this $\walkUnitary$ operation that an element is inserted and some other element is deleted, hence it is sufficient that the dynamic data structure supports replacement of values.


The unitary $\phaseFlipUnitary$ in Equation~\ref{eq:LooksOfPhaseFlipUnitary} adds a negative phase to the marked states and none to the unmarked states, which means $\phaseFlipUnitary$ implements a subroutine that checks whether or not a vertex $V$ is marked by going through its input-query values $S[V]$ and checking if there is a \threeSUM{} solution present in $S[V]$. Currently, there is no known (time) efficient method to implement this subroutine.\footnote{One would require a dynamic data-structure for efficiently answering \threeSUM{} queries, which is not known to exist.}

Instead, suppose that there exists a constant $\alpha>0$ such that there is a subroutine that can solve this structured version of \threeSUM{} on $r$ elements in $O(r^{1-\alpha})$ quantum time. We can now implement $\phaseFlipUnitary$ in the following way. Call the subroutine that is optimal for solving \threeSUM{} on this  of ordered input. The data-structure queries $q_i(S[V])$ to the structured input can be simulated with an $n^{o(1)}$ overhead in time, because the data structure $D(S[V])$ supports efficient data-structure queries. The time complexity of the original Ambainis' walk algorithm then becomes
\begin{equation}
    r \cdot n^{o(1)} + t_1 \cdot t_2 \cdot n^{o(1)} + t_2 \cdot n^{o(1)} \cdot r^{1-\alpha},
\end{equation}
which, after ignoring all the $n^{o(1)}$ factors, becomes
\begin{equation}
   r + t_1 t_2 + t_2 r^{1-\alpha}.
\end{equation}
Substituting the values of $t_1=O(\sqrt{r})$ and $t_2=O((n/r)^{1.5})$ we obtain a total time complexity of order
\begin{equation}
\label{eq:InformalJustBeforeFinal}
   r+ \frac{n^{1.5}}{r} + \frac{n^{1.5}}{r^{1.5}}\cdot n^{1-\alpha}.
\end{equation}
The total time taken in Equation~\ref{eq:InformalJustBeforeFinal} roughly becomes 
\begin{equation}
\label{eq:InformalFinal}
   r+ \frac{n^{1.5}}{r} + \frac{n^{1.5}}{r^{1.5}} \cdot r^{1-\alpha}.
\end{equation}
Given that $r=n^\beta$ for a $\beta \in (0,1)$, it is easy to see that for every $0< \alpha <1$, there exists a $\beta$ such that $\max(\frac{1}{2},\frac{1}{2\alpha +1}) < \beta <1$, and then the value of (\ref{eq:InformalFinal}) becomes strictly sublinear. It then follows that there is no sub-linear quantum time algorithm for solving the structured version of \threeSUM{}, unless \QthreeSUMconjecture{} is false.
\end{proof}

\medskip\noindent
We have omitted several details from the above proof sketch. One omission is that we neglected to account for the error (in the quantum walk and in the invoked subroutine for \threeSUM{}). This is simple to account for and we will do so in Section \ref{sec:HardnessOfSortedUniqueSpaceInefficient}. The most crucial omission is that we will actually require \textit{probabilistic} dynamic data-structures in our reductions. Randomness seems to be required because no dynamic sorting data-structure is known that is simultaneously time-efficient, space-efficient, history-independent, and deterministic. However, a solution exists if any of these four requirements is removed. We will first (in Section \ref{sec:HardnessOfSortedUniqueSpaceInefficient}) present a solution which uses a deterministic data-structure, but large space, and then (in Section \ref{sec:HardnessOfSortedUniqueSpaceEfficient}) a probabilistic solution which is also efficient in space. It is an interesting open question in classical data-structures to provide, or disprove the existence of, a dynamic data-structure that simultaneously satisfies all four requirements.

\subsection{Applications}\label{sec:SummaryQuantumSumHardProblems}

We use our proof strategy to show, conditional on \QthreeSUMconjecture{}, tight lower-bounds on several computational-geometry problems, on \ConvolutionThreeSum{} and on \ZeroWtTriangle{} problem. Our lower-bounds show that the quantum speed-up is at most quadratic for all of these problems.

Our lower-bounds on \ConvolutionThreeSum{} and \ZeroWtTriangle{} tightly match the Grover-based speed-up that quantum algorithms can get for these problems.

Our quantum reductions from \threeSUM{} to computational-geometry problems are complementary to a recent paper by Ambainis and Larka~\cite{Ambainis-QuantumGeometryProblems-2020}, where they present quantum speed-ups for several such problems. Our results show, under the \QthreeSUMconjecture{}, that all of the speed-ups obtained by Ambainis and Larka are optimal. There are also computational-geometry problems for which the \QthreeSUMconjecture{} gives us a lower-bound, but for which no quantum speed-up is known.

Table \ref{table:SummaryQuantumSumHardProblems} (in page \pageref{table:SummaryQuantumSumHardProblems}) summarizes our results. It also includes the best-known \textit{classical} upper and lower-bounds.

\subsection{Future directions and open questions}
\label{sec:FutureDirections}

The study of quantum fine-grained complexity is just beginning. Classically, there are many fine-grained reductions laying out the structure of the class \P, but only a few of such reductions have been established for \BQP. This forms an appealing avenue for future work, as not only is the topic very much unexplored, any tight lower-bounds given by quantum fine-grained reductions will allow us to understand how much quantum speed-up is possible.

The following is a non-exhaustive list of questions which are currently open, and which we hope will benefit from the approach contained in our paper:

\begin{itemize}
    \item Table \ref{table:SummaryQuantumSumHardProblems} contains four problems for which we can prove some quantum lower-bound, conditioned on the \QthreeSUMconjecture{}. Is this lower-bound tight, i.e., are there matching algorithms? Or can we prove a higher lower-bound, perhaps based on a different conjecture?
    \item In the classical setting, there are problems, other than \threeSUM{}, which serve as a basis for fine-grained reductions, e.g. the Orthogonal Vectors problem, the all-pairs shortest-path problem \cite{Survey-VVWilliams-15}. What lower-bounds can we prove in the quantum setting, based on these problems? Can we prove tight bounds on quantum speed-ups?
    \item The \CthreeSUMconjecture{} itself gives various other lower-bounds in the classical setting, which we did not study in the quantum setting, namely lower-bounds against dynamic data-structure problems. Can these lower-bounds be proven in the quantum regime, also?
    \item More generally, for what other problems can we prove that the known quantum speed-up is optimal, under a reasonable hardness hypothesis such as the \QthreeSUMconjecture{}?
\end{itemize}

\medskip
The various papers using dynamic data-structures in quantum walks, including \cite{Ambainis-ElementDistinctness-2004,Aaronson-ClosestPair-2019} and our paper, give rise to an interesting question in classical data-structures. The vast majority of space-efficient dynamic data-structures are not history-independent: history-independence is a feature which cannot be properly motivated if one is only interested in classical algorithms, but which is fundamentally necessary for using the dynamic data-structure as part of a quantum walk. One can then attempt to understand for which problems do history-independent, memory and time-efficient dynamic data-structures exists. For sorting, the only known solution (skip lists) is randomized. Is this necessary? More generally, what dynamic data-structure problems have solutions that are simultaneously deterministic, time-efficient, space-efficient, and history-independent? Can we prove lower-bounds against data-structures obeying all four criteria simultaneously, which we cannot prove against data-structures obeying only three among the four criteria?

\subsection{Structure of the paper} 
The structure of the rest of the paper is as follows. In Section~\ref{sec:ModelsOfComputation} we describe our model of computation, and in Section~\ref{sec:OtherVersionsThreeSum} we describe various simple variants of the \threeSUM{} problem and show that the \QthreeSUMconjecture{} is equivalent for these versions. (These are not the structured versions we mentioned earlier, here the proof of equivalence is very simple.)

In Section~\ref{sec:ImplicationsOfSortedUniqueSUM}, using the approach we sketched above (in Section~\ref{sec:OrderingDoesntHelpSUM}), we give a full proof that, under the \QthreeSUMconjecture{}, two ``structured'' variants of \threeSUM{} also require $\Omega(n)$ time on a quantum computer. We give two separate proofs: The first proof (in Section~\ref{sec:HardnessOfSortedUniqueSpaceInefficient}) uses a deterministic data structure which is \emph{space-inefficient}, and the second proof (in Section~\ref{sec:HardnessOfSortedUniqueSpaceEfficient}) uses a probabilistic data structure which is space-efficient. As direct implications of these hardness results, in Section~\ref{sec:AppendixThreeSumHardGeometryProblems} we present conditional quantum time lower bounds for several computational geometry problems.

Lastly, in Section~\ref{sec:TimeLowerBoundsThreeSum}, we present conditional quantum time lower bound for \ConvolutionThreeSum{} and \ZeroWtTriangle{} problems. This requires us to prove, under the \QthreeSUMconjecture{}, that a third ``structured'' variant of \threeSUM{} also requires $\Omega(n)$ time on a quantum computer.

\begin{table}[h]
    \begin{center}
\begin{tabular}{ | m{13em} | m{3em}| m{8em} | m{3em}|} 

\multicolumn{2}{r}{\makecell{\textbf{\threeSUM{}-based quantum$\qquad$}\\\textbf{lower-bounds (our results)} $\quad$\rotatebox[origin=c]{90}{$\dlsh$}$\quad$}} & \multicolumn{2}{r}{\makecell{\textbf{Classical}$\qquad\qquad$\\\textbf{complexity} ($\ast\ast$)$\quad$\rotatebox[origin=c]{90}{$\dlsh$}$\;$}}\\
\hline
\centering \textbf{Problems} & & \makecell{\bf Quantum\\\bf upper-bound}  & \\
\hline
\centering \geomBase{} & \centering $\Omega(n)$ & $\widetilde{O}(n)$ ($\ast$) & $\Theta(n^2)$\\ 
\hline
\centering \threePointsOnLine{} &  \centering $\Omega(n)$ & $O(n^{1+o(1)})$ \cite{Ambainis-QuantumGeometryProblems-2020} & $\Theta(n^2)$ \\
\hline
\centering \PointOnThreeLines{} &  \centering $\Omega(n)$ &  $O(n^{1+o(1)})$  \cite{Ambainis-QuantumGeometryProblems-2020} & $\Theta(n^2)$ \\
\hline
\centering \separator{} & \centering $\Omega(n)$ & $O(n^{1+o(1)})$ \cite{Ambainis-QuantumGeometryProblems-2020} & $\Theta(n^2)$\\
\hline
\centering \stripsCoverBox{} & \centering $\Omega(n)$ & $O(n^{1+o(1)})$ \cite{Ambainis-QuantumGeometryProblems-2020} & $\Theta(n^2)$ \\
\hline
\centering \trianglesCoverTriangle{} & \centering $\Omega(n)$ & $O(n^{1+o(1)})$ \cite{Ambainis-QuantumGeometryProblems-2020} & $\Theta(n^2)$\\
\hline
\centering \PointCovering{}  & \centering $\Omega(n)$ & $O(n^{1+o(1)})$ \cite{Ambainis-QuantumGeometryProblems-2020} & $\Theta(n^2)$ \\
\hline
\centering \VisibilityBetweenSegments{}  & \centering $\Omega(n)$ & $O(n^{1+o(1)})$ \cite{Ambainis-QuantumGeometryProblems-2020} & $\Theta(n^2)$ \\
\hline
\centering \HoleInUnion{} & \centering $\Omega(n)$ & $O(n^{1+o(1)})$ ($\dag$) & $\widetilde{\Theta}(n^2)$\\
\hline
\centering \TriangleMeasure{} & \centering $\Omega(n)$ & Open! & $\Theta(n^2)$ \\
\hline
\centering \VisibilityFromInfinity{} & \centering $\Omega(n)$ &  Open! & $\Theta(n^2)$\\
\hline
\centering \VisibleTriangle{} & \centering $\Omega(n)$ & $O(n^{1+o(1)})$ ($\dag$) & $\Theta(n^2)$ \\
\hline
\centering \PlanarMotionPlanning{} & \centering $\Omega(n)$ & Open! & $\Theta(n^2)$ \\
\hline
\centering \threeDmotionPlanning{} & \centering $\Omega(n)$ & Open! & $\Theta(n^2)$ \\
\hline
\centering \GeneralCovering{} & \centering $\Omega(n)$ & $O(n^{1+o(1)})$ \cite{Ambainis-QuantumGeometryProblems-2020} & $\Theta(n^2)$
\\
\hline
\centering \ConvolutionThreeSum{} & \centering $\Omega(n)$ & $O(n)$ ($\ast$) & $\Theta(n^2)$ \\
\hline
\centering \ZeroWtTriangle{} & \centering $\Omega(n^{1.5})$ & $O(n^{1.5})$ ($\ast$) & $\Theta(n^3)$ \\
\hline
\end{tabular}
\end{center}
\ \\
\begin{itemize}
    \item[($\ast$)] Using a simple Grover speed-up on the classical algorithm.

    \item[($\dag$)] Implicit in \cite{Ambainis-QuantumGeometryProblems-2020}, by using the classical reduction to \trianglesCoverTriangle{} and then using the corresponding quantum algorithm.
    
    \item[($\ast\ast$)] All upper-bounds are straightforward: For problems like \ConvolutionThreeSum{} and \ZeroWtTriangle{} the best known algorithms use brute force, for the computational-geometry problems, the upper-bounds follow from geometry arguments \cite{Overmars-ComputationalGeometry-1995}. All lower-bounds for computational geometry problems are from \cite{Overmars-ComputationalGeometry-1995}, the lower-bound for \ConvolutionThreeSum{} follows from \cite{Patrascu-Convoluted3SUM-2010}, and, the lower-bound for \ZeroWtTriangle{} follows from \cite{Williams-FindingTypesOfTriangles-2009}.
\end{itemize}

\caption{This is a summary of all the Quantum-\threeSUM-hard problems mentioned in this paper, with (almost) matching upper bounds for most of them.}

\label{table:SummaryQuantumSumHardProblems}
\end{table}

\newpage
\section{Preliminaries}
\label{sec:Preliminaries}

\subsection{Model of Computation: Standard Quantum Circuit Model augmented with Random Access Gates}
\label{sec:ModelsOfComputation}

We assume that the input is given as a black-box and can be accessed in superposition. One application of this black-box (also known as an oracle) is called a \textit{query}. An input to the oracle consists of two registers: An index\footnote{The query access could be of various forms, for example: In a problem where the input is a graph, the oracle gives query access to the adjacency matrix of the graph, which is indexed by two variables instead of one.} $i$ and space for the description ($d_i$) of the object that is to be queried. The oracle $\mathcal{O}$ is then the unitary given by $\mathcal{O}\ket{i}\ket{b}= \ket{i}\ket{(b + d_i) \mod D}$, where $D$ is the maximum value of the range of values that the description $d_i$ of the object takes. If there is a state $\ket{\psi}=\sum_{i}\alpha_i\ket{i}\ket{0}$ then $\mathcal{O}\ket{\psi}=\sum_i \alpha_i \ket{i}\ket{d_i}$. 

For the purpose of this paper, we are interested in the time complexity of quantum  algorithms, where we define the time complexity as the total number of elementary gates in a circuit implementing the algorithm. By elementary gates we mean all the one-qubit quantum unitaries, the two-qubit CNOT gate \cite{Smolin-ElementaryGates-95} and the multi-qubit random access gates (RAGs) which we define below. Note that any circuit with $s$ number of CNOTS and single-qubit unitaries can be implemented using a discrete gate set comprising of Hadamard, CNOT, phase and $\pi/8$ gates to an accuracy $\epsilon$ with $s\log^c(s/\epsilon)$ gates for a constant $c\approx 2$ \cite{Kitaev-ContinousToDiscreteGateSet-1997}.

A random access gate (RAG) is defined as follows: It takes three inputs, $\ket{i}, \ket{b}, \ket{z}$ for an $i\in [m]$, where $i$ is a register of $\log m$ qubits, $b$ is a single qubit and $z$ is a register of $m$ qubits; and it implements the mapping:
\begin{equation*}
    \ket{i,b,z} \mapsto \ket{i, z_i, z_1...z_{i-1}bz_{i+1}...z_m}
\end{equation*}

The usual quantum circuit model does not include RAGs. Adding RAGs is necessary in our case because, without such gates, even simple data-structure operations, that would take logarithmic time on a classical random-access machine, would take polynomial time in the quantum model. 

The inclusion of RAGs makes the usual quantum circuit model equivalent, up to logarithmic factors, to the time complexity of quantum random-access machines (QRAMs), in that time for QRAMs corresponds to the number of gates in the circuit, and memory for QRAMs corresponds to the number of wires (i.e. the number of qubits, the dimension of the Hilbert space). It is also not hard to see that RAGs can be implemented using a $O(\log n)$-depth, $O(n)$-size ``parallel'' quantum circuit comprising only of CNOT and single-qubit gates. RAGs were used in the sub-linear quantum-time algorithm for the Element Distinctness problem presented by \cite{Ambainis-ElementDistinctness-2004}, and in the sub-linear quantum-time algorithm for the closest-pair problem appearing in \cite{Aaronson-ClosestPair-2019}.

\subsection{Simple variants of \threeSUM{}}
\label{sec:OtherVersionsThreeSum}

The standard \threeSUM{} problem is defined as follows: Given a list $S$ of $n$ integers, do there exist elements $a, b, c \in S$ such that $a+b+c=0$?
Furthermore, there are several other variants of the \threeSUM{} problem that have been useful as intermediary steps for reductions in the classical case. We will begin by considering the following two simple variants:
\begin{enumerate}
    \item \threeSUMpatrascuVersion{}: Given a list $S$ of $n$ integers, are there $a, b, c \in S$ such that $a+b=c$?
    \item \threeSUMthreeList{}: Given three lists $A, B, C$ of $n$ integers each, are there $a \in A$, $b \in B$ and $c \in C$ such that $a+b=c$?
\end{enumerate}

We now show that \threeSUM{}, \threeSUMpatrascuVersion{}, and \threeSUMthreeList{} can all be quantumly reduced to each other with $O(\sqrt{n})$ pre-computation time, followed by an \textit{on-the-fly} fast reduction for any query. Meaning, the reduction is given input $X$ and outputs $Y$ in the following sense: after some pre-computation on $X$, obtaining any of the integers in the list or lists in $Y$ can be done in $O(1)$ time (``on-the-fly'') by querying $X$. This establishes, therefore, that the \QthreeSUMconjecture{} can be equivalently stated for any of these simple variants.

The reduction from \threeSUM{} to \threeSUMthreeList{} is simple, and requires no pre-computation.
We set $A=S, B=S$ and $C=-S$, and now $a\in A, b\in B, c \in C$ have $a+b=c$ if and only if $a, b, (-c) \in S$ have $a+b+(-c)=0$.
The reduction from \threeSUMpatrascuVersion{} to \threeSUMthreeList{} is also simple and on-the-fly with no required pre-computation. Simply set $A=S$, $B=S$ and $C=S$. 

The reduction from \threeSUMthreeList{} to \threeSUMpatrascuVersion{}, is slightly more complicated and is almost identical to the reduction presented in Theorem~3.1 by \cite{Overmars-ComputationalGeometry-1995}, and is as follows: As a pre-computation step, compute the element $m=2\max (A,B,C)$. This part takes $O(\sqrt{n})$ quantum time. Create a list $S$ of size $3n$: For each $a \in A$ put $a'=a+m$ in $S$, for each $b \in B$ put $b'=b+3m$ in $S$, and, for each $c \in C$ put $c'=c+4m$ in $S$. Clearly, if
$a + b = c$ then $a' + b' = c'$. Without loss of generality one can assume the elements of the lists $A,B,C$ are strictly positive because one can add a big number $k$ to all the elements in lists $A,B$ and $2k$ to the elements in list $C$. Additionally, with some elementary calculations one can easily see that whenever there are three elements in $S$ such that $a'+b'=c'$, the corresponding $a,b,c$ come from three different sets $A,B,C$, respectively.

The reduction from \threeSUMthreeList{} to \threeSUM{} is very similar to the above, and is given in Theorem~3.1 by \cite{Overmars-ComputationalGeometry-1995}. We will not repeat it here. As above, the reduction contains a pre-computation step where the maximum of all the lists $A, B, C$ is computed, making the quantum reduction take $O(\sqrt{n})$ pre-computation time. Thereafter, it is an on-the-fly fast reduction for any query. 
 
Hence it follows that the \QthreeSUMconjecture{} is equivalent to the same conjecture stated for \threeSUMpatrascuVersion{}, or for \threeSUMthreeList{}.

\section{Lower-bounds for two structured versions of \threeSUM{}}
\label{sec:ImplicationsOfSortedUniqueSUM}
In 1995, Gajentaan and Overmars \cite{Overmars-ComputationalGeometry-1995} showed that \threeSUM{} can be reduced to several computational geometry problems, proving that these problems cannot have truly sub-quadratic classical algorithms unless the \CthreeSUMconjecture{} is false. These results are proven by first exhibiting a reduction from \threeSUM{} to some fundamental computational geometry problems (\sortedGeomBase{} and \threePointsOnLine{}), and then constructing further reductions among computational geometry problems. See Figure \ref{fig:GeometryProblemsRelations} for an overview of such reductions.

The reductions among computational geometry problems are all simple to adapt to the quantum setting. We will do so in Section \ref{sec:AppendixThreeSumHardGeometryProblems}. However, the fundamental reduction to \sortedGeomBase{} requires sorting the \threeSUM{} instance, and the fundamental reduction to \threePointsOnLine{} requires both sorting and removing duplicate elements. This means that we cannot trivially adapt these reductions to the quantum setting. We overcome this obstacle by employing dynamic data structures, as sketched in the proof Theorem \ref{thm:HardnessOfOrderedSUM}. We will now show how we can dynamically maintain a list in sorted order and free of duplicate elements, by way of a history-independent dynamic data-structure.

\begin{figure}[h]
    \centering
    \includegraphics[scale=0.49]{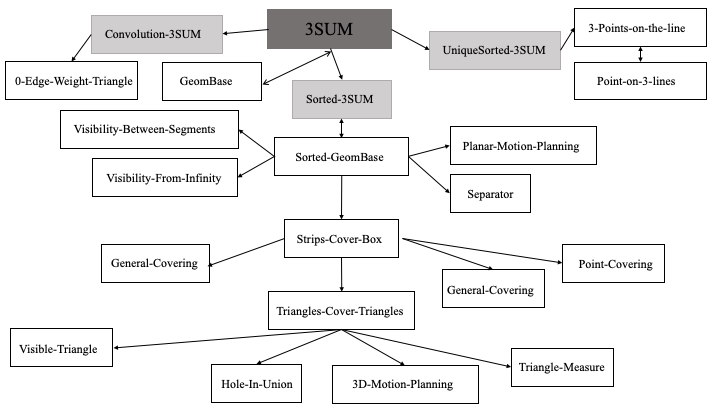}
    \caption{Overview of the different reductions between \threeSUM{}, \ConvolutionThreeSum{}, \ZeroWtTriangle{} and some Computational Geometry problems. The same reductions can be shown both classically and quantumly, but the classical reductions from \threeSUM{} to \ConvolutionThreeSum{}, \sortedThreeSUM{}, and \uniqueSortedSUM{} cannot be trivially translated to the quantum setting.}
    \label{fig:GeometryProblemsRelations}
\end{figure}

\subsection{Hardness of \sortedThreeSUM{} and \uniqueSortedSUM{} using Space Inefficient Data Structures}
\label{sec:HardnessOfSortedUniqueSpaceInefficient}
In this section we present an example of a dynamic data structure that is deterministic and history-independent and we make use of this data structure along with the quantum walk algorithm to prove hardness results for \sortedThreeSUM{} and \uniqueSortedSUM{} problems based on the conjectured hardness of the \threeSUM{} problem.

\paragraph{Example of such a data structure.} Let $M = \{-n^3,\ldots, n^3\}$, and let $S'=(x_1,x_2,...,x_n) \in M^n$ denote an input to the \threeSUM{} problem.

\begin{figure}[h]
    \centering
    \includegraphics[scale=0.51]{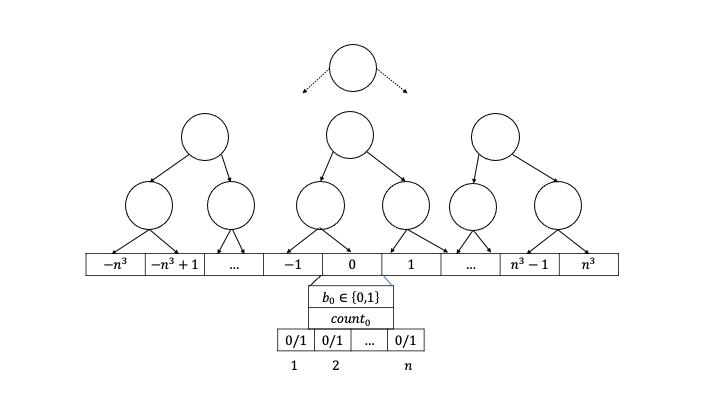}
    \caption{Each of the circular nodes contains two variables \countChildren{} and \countUniqueChildren{} that store the number of marked children, total and unique, respectively. Each of the elements in the bit vector contains three variables: The first indicates if an element is present, the second indicates how many times it is present, and the third variable is a bit vector of length $n$ to indicate the corresponding indices to the element.}
    \label{fig:PrefixTreeBitVector}
\end{figure}

Consider the following bit vector of length $|M|$ with a binary prefix tree built on top of the bit vector as shown in Figure~\ref{fig:PrefixTreeBitVector}. Each element of the bit vector, i.e., the leaf nodes that are indexed by $i \in M$, store the following information: (1) A corresponding bit value $b_i \in \{0,1\}$ where $b_i=1$ if $i \in S'$, also referred as a marked leaf node, otherwise $b_i=0$ referred to as unmarked. (2) A counter ($count_i$) that stores the number of times the element $i$ occurs in $S'$. (3) Lastly, a bit vector on $[n]$ that marks all the indices $j$ for which the corresponding element $x_j=i$. Every other node in the tree has two paths labelled by $0$ or $1$. The path starting at the root node and ending at a leaf node indexed by $i$ is labelled by binary representation of $i$ prepended with a bit that indicates the sign, $0$ for the minus sign and $1$ indicating the positive sign. Nodes other than the leaf nodes have two counters. The first counter stores the number of marked descendant leaf nodes, a variable we will denote by \countChildren, and the second counter stores the number of marked unique descendant leaf nodes, which we will denote by \countUniqueChildren. We claim that the following operations can be implemented efficiently using this data structure. 

\begin{enumerate}
    \item \textbf{Insertion and Deletion.} If an element $x_j$ with $x_j=i$ needs to be inserted then its corresponding bit value $b_i$ needs to be set to $1$ (if it is not already $1$), $count_i=count_i+1$ and the index $j$ in the bit vector for $i$ needs to be marked. All of this can be achieved in $O(1)$ amount of time. There after, the variables \countChildren{} and \countUniqueChildren{} of all its ancestral nodes up to the root node needs to be incremented, which can be done in $\log |M|$ time by traversing the path from root node to the leaf node indexed by $i$. The procedure for deleting an element is just the reverse of the insertion procedure. 
    \item \textbf{Indexing.} This can be requested in two ways: (1) For some $k \in [|S'|]$ return the $k\pth$ largest element of the set~$S'$, or, (2) for some $\bar{k} \in [O(|S'|)]$, return $\bar{k}\pth$ largest unique value of the set $S'$. For example, if $S'=[1, 3, 3, 3, 3, 4, 5]$, then the $4\pth$ largest element of $S'$ is $3$, while the $4\pth$ largest unique value of $S'$ is $5$. For this data structure, both these types of indexing can be implemented efficiently in the following way: Set a pointer to the root node. Clearly the value of \countChildren{} (or \countUniqueChildren) is equal to the total number of marked (unique) leaf nodes. Let $k_{\textit{left}}$ denote the count on the left child of the pointer and similarly let $k_{\textit{right}}$ denote the count on the right child of the pointer. If $k>k_{\textit{left}}$, then set the pointer to the right child and $k=k-k_{\textit{left}}$, and repeat the process recursively with the new $k$. If $k\leq k_{\textit{left}}$ then set the pointer to the left child and repeat the process recursively with the same $k$. Repeat until you reach a marked (unique) leaf node. This entire process takes $\log |M|$ amount of time.
    \item \textbf{Finding.} One can in $O(1)$ time find whether an element $i$ is marked, i.e.~whether it belongs to the set $S'$, by checking if $b_i=1$. 
\end{enumerate}

Whereas the solution above is efficient in terms of the number of gates, it is not efficient in terms of the number of necessary qubits. Indeed, we need about $O(n^4)$ qubits of memory to implement this data structure. As our computational model is the standard quantum circuit model augmented with RAGs (cf.~Section~\ref{sec:ModelsOfComputation}), having a space-inefficient data structure does not affect the runtime of our algorithm. 

\paragraph{The \QthreeSUMconjecture{} for $\sortedThreeSUM$.} The data structure in the above example maintains the input in a sorted order. Furthermore, it is history-independent, as required by Theorem \ref{thm:HardnessOfOrderedSUM}. Let us then define \sortedThreeSUM{} to be equal to \threeSUM{}, with a promise that the input list is sorted. Hence we may use Theorem~\ref{thm:HardnessOfOrderedSUM} to prove the following statement.

\begin{thm}
\label{thm:HardnessOfSortedSUMInefficientDS}
If there is a bounded-error quantum algorithm that solves \sortedThreeSUM{} in $\widetilde{O}(n^{1-\alpha})$ time for some constant $\alpha > 0$ using $g(n)$ qubits of memory, then there exists a constant $\delta>0$ such that \threeSUM{} can be solved in $\widetilde{O}(n^{1-\delta})$ quantum time with probability $1-o(1)$ using at most $O(n^4)+g(n^{1-O(1)})$ qubits of memory.
\end{thm}

As the used data-structure is deterministic and history-independent, the sketch of the proof of Theorem \ref{thm:HardnessOfOrderedSUM} is quite complete, and we need only account for the error in the algorithm. The only additional source of error, other than the error of $O(1/\poly(n))$ induced by the quantum walk-based query algorithm (see, e.g., Equation~56 in \cite{Andrew-SubsetFinding-2005}), comes from our invocation of the bounded-error quantum algorithm that solves \sortedThreeSUM{}. If this algorithm succeeds with error probability $\leq 1/3$, we can cheaply reduce the error probability to any small $\varepsilon >0$ by running the subroutine $O(\log (1/\varepsilon))$ times and taking the majority. We would like to chose $\varepsilon$ such that $(1-\varepsilon)^{t_2}=1-o(1)$, where $t_2$ is the number of times the subroutine is called (it is the same parameter as the in the proof sketch of Theorem \ref{thm:HardnessOfOrderedSUM}). Given that $t_2=(n/r)^{1.5}$ with $r$ polynomially related to $n$, it therefore suffices to choose $\varepsilon=(1/\poly(n))$, which means that the total number of times we will invoke the \sortedThreeSUM{} subroutine to compute the original \threeSUM{} problem is only worse by a factor of $O(\log n)$. This is satisfactory, since we ignore all poly-logarithmic factors of $n$ in our analysis.

Apart from the expensive (in terms of space) data structure, the rest of quantum walk algorithm uses only a poly-logarithmic number of qubits, hence, at most $O(n^4 + m')$ qubits suffice to implement this algorithm, where $m'$ denotes the number of qubits required by the subroutine for \sortedThreeSUM{}.

An immediate implication of  Theorem~\ref{thm:HardnessOfSortedSUMInefficientDS} is that a sublinear quantum time algorithm for \sortedThreeSUM{} would imply a sublinear quantum time algorithm for \threeSUM{}, and therefore would contradict the \QthreeSUMconjecture{}.

\begin{corollary}\label{cor:HardnessSortedSUM}
$\sortedThreeSUM$ can be solved in time $\tilde O(n^{1-\eps})$, for some $\eps > 0$, if and only if the \QthreeSUMconjecture{} is false.
\end{corollary}

\paragraph{Lower bounds conditioned on hardness of \sortedThreeSUM{}.} Employing the results of Corollary \ref{cor:HardnessSortedSUM} we are able to show that conditioned on the \QthreeSUMconjecture{}, the following computational geometry problems and many more (refer to Figure~\ref{fig:GeometryProblemsRelations}) require $\widetilde{\Omega}(n)$ quantum time.
\begin{enumerate}
    \item The \separator{} problem: Given a set of $n$ (possibly half-infinite) closed horizontal line segments, is there a non-horizontal separator?
    \item The \stripsCoverBox{} problem: Given a set of strips in the plane does their union contain a given axis-parallel rectangle?
    \item The \trianglesCoverTriangle{} problem: Given a set of triangles in the plane, does their union contain another given triangle?
\end{enumerate}
Having established the quantum hardness of the \sortedThreeSUM{} problem, we can now directly reduce \sortedThreeSUM{} to these problems by a simple adaptation to the quantum setting of the classical reductions presented in \cite{Overmars-ComputationalGeometry-1995}.

\paragraph{The \QthreeSUMconjecture{} for $\uniqueSortedSUM$.}
In addition to storing the input in a sorted order, the above data structure also allows efficient access to the sorted list of unique elements (indexing of type (2)). Let us then define \uniqueSortedSUM{} to be equal to \threeSUM{}, with a promise that the input list is sorted and all its elements are distinct. It then follows from Theorem \ref{thm:HardnessOfOrderedSUM}:

\begin{thm}
\label{thm:HardnessOfUniqueSortedSUMInefficientDS}
If there is a bounded-error quantum algorithm that solves \uniqueSortedSUM{} in $\widetilde{O}(n^{1-\alpha})$ time for some constant $\alpha > 0$ using $g(n)$ qubits of memory, then there exists a constant $\delta>0$ such that \threeSUM{} can be solved in $\widetilde{O}(n^{1-\delta})$ quantum time with probability $1-o(1)$ using at most $O(n^4)+g(n^{1-O(1)})$ qubits of memory.
\end{thm}

\medskip\noindent
As above, it follows as a corollary that the \QthreeSUMconjecture{} can be stated equivalently for $\uniqueSortedSUM$:
\begin{corollary}\label{cor:HardnessUniqueSortedSUMInefficientDS}
$\uniqueSortedSUM$ can be solved in time $\tilde O(n^{1-\eps})$, for some $\eps > 0$, if and only if the \QthreeSUMconjecture{} is false.
\end{corollary}

\paragraph{Lower bounds conditioned on hardness of \uniqueSortedSUM{}.} As implications to Corollary \ref{cor:HardnessUniqueSortedSUMInefficientDS} we show that conditioned on the \QthreeSUMconjecture{}, the following computational geometry problems require $\widetilde{\Omega}(n)$ quantum time.
\begin{enumerate}
    \item The \threePointsOnLine{} problem: Given a set of points in the plane, is there a line that contains at least three of the points?
    \item The \PointOnThreeLines{} problem: Given a set of lines in the plane, is there a point that lies on at least three of them?
\end{enumerate}
Both of these problems are computationally equivalent, as the second problem is the exact dual of the first problem under the Point-Line dualization. The classical reduction from \threeSUM{} to these two problems assumes that the input to \threeSUM{} is unique, i.e.\ there are no duplicate elements in the input. As discussed earlier in Section \ref{sec:OrderingDoesntHelpSUM}, the \CthreeSUMconjecture{} also trivially holds for this promise version of \threeSUM{}, but such a claim can not be easily made in the quantum setting. Therefore, we use the results of Corollary~\ref{cor:HardnessSortedSUM} and Corollary~\ref{cor:HardnessUniqueSortedSUMInefficientDS} to establish that both \uniqueSortedSUM{} and \threeSUM{} are equally hard as the original \threeSUM{} problem in the quantum setting as well.

We give an illustration of the relations between the different geometry problems in Figure~\ref{fig:GeometryProblemsRelations} and we point the readers to Section~\ref{sec:AppendixThreeSumHardGeometryProblems} for details of some of these reductions.

\subsection{Hardness of \sortedThreeSUM{} and \uniqueSortedSUM{} using Space Efficient Data Structures}

\label{sec:HardnessOfSortedUniqueSpaceEfficient}
Recall the statement of Theorem~\ref{thm:HardnessOfOrderedSUM} in Section~\ref{sec:OrderingDoesntHelpSUM}. This theorem states that, under the \QthreeSUMconjecture{}, structured versions of \threeSUM{} require linear quantum time given the existence of a particular data structure. Apart from the requirement that the operations on the data structure are efficient, there are two other requirements: Firstly that the data structure has a unique, history-independent representation in memory, and secondly that every data structure operation terminates within a fixed amount of time $t = n^{o(1)}$. We would now like to improve the results from the earlier subsection by using \emph{space-efficient} data structures to prove conditional hardness of the \sortedThreeSUM{} and \uniqueSortedSUM{} problems.
However, no deterministic space-efficient data structures are known, and in fact the known probabilistic data structure for sorting (skip lists) no longer satisfies the termination condition, meaning, there will exist some ``bad'' inputs for which some of the data-structure operations take too much time. In this subsection we will  show that the claims of Theorem~\ref{thm:HardnessOfSortedSUMInefficientDS} hold even after a relaxation towards using probabilistic data structures, provided that the probability that an input is ``bad'' is small.

The approach is the same as before: we assume that there are sublinear-time quantum algorithms that solve \sortedThreeSUM{} or \uniqueSortedSUM{}, and we obtain a sublinear-time quantum-walk-based algorithm for \threeSUM{}

\paragraph{Recap of proof strategy of Theorem \ref{thm:HardnessOfSortedSUMInefficientDS} from Section~\ref{sec:HardnessOfSortedUniqueSpaceInefficient}} To prove our result, we use the quantum walk algorithm presented by Ambainis that solves the Element Distinctness problem using an optimal number of queries and in an optimal amount of time \cite{Ambainis-ElementDistinctness-2004} (up to poly-logarithmic factors). As shown by \cite{Andrew-SubsetFinding-2005}, this algorithm can be used to solve the \threeSUM{} problem in $\Theta(n^{3/4})$ queries, with a matching lower bound given by Belovs and \v Spalek \cite{Belovs-KSUMLowerBound-2012}. This optimal query algorithm for \threeSUM{} is a quantum walk on the \textit{Johnson graph} ($J(n,r)$ with $r=n^{3/4}$), for which the high-level idea is as follows: Consider a graph of $n \choose r$ vertices with each vertex of the graph labelled by an $r$-sized subset of $[n]$. Put an edge between two vertices if and only if these two sets differ by exactly two elements. The optimal query algorithm is a quantum walk on this resultant graph $J(n,r)$. In order to minimise the number of queries required for solving \threeSUM{}, the key idea is to store the list of query values of every $r$-sized subset along with the state that represents that subset. Having done that, as a part of the quantum walk algorithm, a subroutine checks if there is a \threeSUM{} solution in the queried values corresponding to each of these $r$-sized subsets, but because it can check in superposition there is a quantum speedup that is guaranteed. While this checking step requires no additional queries, actual implementation of such a subroutine does require a significant amount of run time, so that the quantum walk algorithm for \threeSUM{} takes at best linear time. It is this subroutine, i.e.~the one that checks for a \threeSUM{} solution in the $r$-sized set of queries, that we would like to further speed up. To do so we use two components: First, a dynamic data structure that maintains a sorted list with additional support of efficient (i.e.~$n^{o(1)}$) insertions, deletions, look-ups and indexing, and secondly, the assumed sublinear time algorithm that solves \sortedThreeSUM{} (or \uniqueSortedSUM{}). 

What differs from the results presented earlier in Section~\ref{sec:HardnessOfSortedUniqueSpaceInefficient} is that the data structure we use here is  \emph{space efficient} and \emph{probabilistic}. In spite of that, the probability of error is small and does not affect the quantum walk algorithm significantly. 

Please note that, the proof of this theorem can only be understood if the reader is familiar with the results from Section~6 in \cite{Ambainis-ElementDistinctness-2004}. 

\begin{thm}
\label{thm:HardnessOfThreeSumProbabilisticDataStructure}
If there is a bounded-error quantum algorithm that solves \sortedThreeSUM{}, i.e.~\threeSUM{} with a promise that the input is sorted, in $\widetilde{O}(n^{1-\alpha})$ time for some constant $\alpha >0$ using $g(n)$ qubits of memory, then there exists constants $\beta < 1,\delta>0$ such that \threeSUM{} can be solved in $\widetilde{O}(n^{1-\delta})$ quantum time with high probability using at most $\widetilde O(n^\beta) + g(n^\beta)$ qubits of memory.
\end{thm}

\begin{proof}
We first prove our result for the case where there is at most one solution to the \threeSUM{} problem. Our results can be generalised to the situation where there are multiple solutions by running our single-solution algorithm repeatedly on different sized subsets of input $x_i, i \in [n]$, as in Algorithm~3 in Section~5 of  \cite{Ambainis-ElementDistinctness-2004}. The analysis of the multi-solution algorithm given by \cite{Ambainis-ElementDistinctness-2004} holds in our situation as well.

\begin{figure}[t]
    \centering
    \includegraphics[scale=0.51]{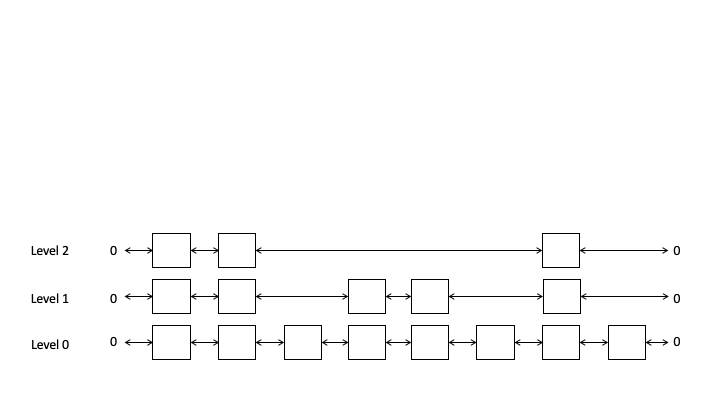}
    \caption{Illustration of a three-level skip list}
    \label{fig:SkipList}
\end{figure}

As the first part of the proof is identical to the proof sketch of Theorem~\ref{thm:HardnessOfOrderedSUM}, we directly jump onto discussing the data structure. 

\paragraph{Data structure.}As mentioned above, the proof of this theorem follows the same general approach as to the quantum walk algorithm by \cite{Ambainis-ElementDistinctness-2004, Andrew-SubsetFinding-2005}. The key idea of the quantum walk algorithm is to use an $r$-sized set of indices from $[n]$ along with their respective query values as nodes of the quantum walk. Let us denote this set, which will describe the data we want to store in the data structure, by $S'$.

Consider the data structure used in Section $6$ of the quantum walk algorithm for solving the Element Distinctness problem (on $n$ elements) by \cite{Ambainis-ElementDistinctness-2004}. This data structure supports efficient insertions, deletions and finding. Furthermore, elements stored in this data structure are of the type $(i,x_i)$, where $i$ is an index and $x_i$ denotes the query value corresponding to index $i$ with elements first sorted by the value of $x_i$ and then sorted by the value of $i$. 

We use the same data structure as the one presented by Ambainis \cite{Ambainis-ElementDistinctness-2004}, a combination of a hash table and a skip list, extended to also allow for efficient indexing.

\begin{enumerate}
    \item \textbf{Hash table.} The hash table consists of $r$ buckets each of which is equipped to store a maximum of $\ceil{\log n}$ entries. We will soon see how each entry uses $O(\log^2 n + \log m)$ qubits, here $m$ is the upper bound on number of qubits required to store the query values. The total memory used is therefore $O(r\log^3(n+m))$.
    
    Whenever there is a request for inserting a new element, say $(i,x_i)$, a memory location is allocated for $(i,x_i)$ using the following hash function on the value of $i$: 
    \begin{equation}
    \label{eq:HashTableForUniquenessCondt}
        h(i)=\floor{r \cdot i/n} + 1. 
    \end{equation}
    This is done to ensure that the data structure always stores the same set $S'$ in the same way in the memory, no matter how it was created (satisfying the \textit{history-independence} condition mentioned in the statement of Theorem~\ref{thm:HardnessOfOrderedSUM}).
    
    The entry for $(i,x_i)$ contains enough memory to store the element $(i,x_i)$ together with storing $\ceil{\log n}+1$ pointers that will be used for the skip list as shown in the Figure \ref{fig:SkipList}.
    
    \item \textbf{Skip list.} Once the memory is allocated for the new element, the node referring to this element has to be included in a skip list pointer structure with $\ceil{\log n}+1$ levels -  Figure \ref{fig:SkipList}. Every element $(i,x_i)$ has a randomly assigned level $l_i$ between $0$ and $\ceil{\log n}$. The element $(i,x_i)$ with level value $l_i$ will be present on all the levels $l$ with $l\le l_i$. The elements in each level of skip list are stored in the order of increasing $x_i$. If there are several $i$ with same $x_i$ then they are ordered by $i$.
    
    \item Additionally, for the data structure operations to be efficient, it is required that the probability for any element to be present in a level $l$ is $1/2^l$. This ensures that with very high probability, no more than $O(\log n)$ elements in any level are encountered while inserting/deleting or finding an element in the skip list. For the details of this analysis, refer to Section $6$ of \cite{Ambainis-ElementDistinctness-2004}.
 
\end{enumerate}
However, this skip list is not yet equipped with an \emph{efficient indexing}\footnote{The term \textit{indexing} in context of a data structure is mentioned in the example presented in Section \ref{sec:ImplicationsOfSortedUniqueSUM}. }, as required by our reduction. We therefore make the following modifications.

\paragraph{Modifications to Ambainis' data structure in order to incorporate efficient indexing.} The type of indexing that we are interested in is as follows: Given some $k \in [|S'|]$ return the $k\pth$ largest element of the set $S'$, i.e.~the element which has the $k\pth$ largest query value among the elements in $S'$, in $n^{o(1)}$ time. This is required for the following reason: Recall that as a part of our reduction, we invoke the subroutine that presumably solves \sortedThreeSUM{} in $O(r^{1-\alpha})$ (for some $\alpha > 0$) quantum time on an $r$-sized subset of $[n]$.
We need to make sure that implementing oracle access to the inputs to this subroutine isn't too costly.
With a few additions to the skip list data structure used by Ambainis, we are able to implement efficient (i.e.~$n^{o(1)}$ time) indexing.

The changes we make to the data structure are as follows: Let $\textit{node}_i$ denote the node that stores the element $(i,x_i)$. Additional to that, the node also stores $O(\log n)$ pointers that are used in the skip list. We add a total of $2\ceil{\log n}$ new variables, with two count variables associated with each level except for the lowest level. For every level that the node is in, the number of elements present between this node and the nodes to its left and right are stored, thereby using two variables for this purpose. If at any level there are no nodes on the left (right) of $\textit{node}_i$ then it stores the number of elements present before (after) this element $(i,x_i)$ in the skip list at level 0. 

We will first analyze how much additional time is required to update these variables, and the variables respective to its neighbouring nodes, when a node is inserted or deleted.
\begin{enumerate}
    \item \textbf{Insertion.} Any new node is inserted using the same procedure as in \cite{Ambainis-ElementDistinctness-2004} which is as follows: Start with the pointer at the top-most level of the skip list, find the (\textit{left} and \textit{right}) positions of the nodes between which the new node can be inserted. Insert it in this level if the level assigned\footnote{The level for each node is assigned using a family of $d$-wise independent hash functions with $d=4\log n+1$ but this $d$ can be set to $c_1 \log n$ for any $c_1$ by using a bigger ($O(n^{\ceil{d/2}})$) family of $d$-wise independent functions which exists. Refer to Theorem~1 by \cite{Ambainis-ElementDistinctness-2004} for more information on this family of hash functions.} to the node agrees with such an insertion otherwise, move to the next level starting from the current position. Repeat this process till the lowest level is reached. 
    
    Consider a scenario where we are investigating the position of the new node on a level $l$. The node could either get inserted at this level or not. The rules of updating the $l$-level counters for neighbouring nodes and for \nodeSkipList{i} depend on whether or not the insertion happens, which is why we consider these two scenarios separately.
    
    \begin{enumerate}
        \item If \nodeSkipList{i}  doesn't get inserted at level $l$, then right (left) $l$-level variable of the left (right) neighbouring node has to be incremented by $1$, while the $l$-level variables of $\textit{node}_i$ can be assigned a $-1$.
        \item The process to update these $l$-level variables when \nodeSkipList{i} gets inserted at level $l$ is slightly trickier. Note that, if the new node does get inserted at level $l$, all the $t$-level variables with $t\leq l$ also need to be updated for \nodeSkipList{i} as well as for its left and right neighbours. For that, we use the following procedure.
     \end{enumerate}   
    \paragraph{Recursive procedure to update $l$-level variables.} Let's say \nodeSkipList{i} is to be inserted at level $l$ and therefore at all the subsequent levels $t$ where $t \leq l$. The insertion procedure, starting from level $l_{max}=\ceil{\log n}+1$ finds the position i.e.~the left and right neighbours of \nodeSkipList{i} at level $l_{max}$, having found that, starting from (what could have been) the left most neighbour of \nodeSkipList{i} at level $l_{max}$ we then traverse along level $(l_{max}-1)$ pointer to find \nodeSkipList{i}'s immediate neighbours at level $(l_{max}-1)$. We repeat this process until we reach the lowest level. Computation of the level variables also uses a similar approach. For any $t>0$, the $t$-level values of each of these neighbouring nodes of \nodeSkipList{i} cannot be computed until we reach the lowest level. Therefore, to achieve this, we use the pointers to these (left and right) nodes at every level during insertion. Once we know the position of the new node on the lowest level, we use a recursive procedure (that we will soon discuss) to compute the values of the count variables for \nodeSkipList{i} and its neighbouring nodes at every level, and update them accordingly.
    
    Let $\nodeNeighbourLeft{i}{}$ denote the left neighbour of \nodeSkipList{i} and similarly let $\nodeSkipList{i}\rightarrow \textit{right}$ denote its right neighbour. The skip list data structure allows us to access these neighbours in constant time. We extend this notation by adding a value on the arrow to indicate the neighbours of \nodeSkipList{i} on  a particular level. For example $\nodeSkipList{i} \xrightarrow[]{l} \textit{left}$ denotes the left neighbour of \nodeSkipList{i} at level $l$, similarly $\nodeSkipList{i} \xrightarrow[]{l} \textit{right}$ denotes its right neighbour at level $l$. Additionally, when \nodeSkipList{i} appears before \nodeSkipList{j} in the skip list, we denote that by $\nodeSkipList{i} <_{order} \nodeSkipList{j}$.   
    
    Furthermore, let $d(\textit{node}_i, \textit{node}_j)$ denote the distance between the two nodes indexed by $i,j$, i.e., the number of elements present between the two nodes at level $0$. The distance function $d(\cdot, \cdot)$ has the following properties:
    \begin{enumerate}
        \item For any node \nodeSkipList{i}, $d(\nodeSkipList{i}, \nodeSkipList{i})=0$.
        \item For any two neighbouring nodes \nodeSkipList{i} and \nodeSkipList{j} at level $0$ the $d(\nodeSkipList{i}, \nodeSkipList{j})=0$. 
        \item For any two nodes \nodeSkipList{i} and \nodeSkipList{j} we have $d(\nodeSkipList{i}, \nodeSkipList{j})=d(\nodeSkipList{j}, \nodeSkipList{i})$.
        \item For any three consecutive nodes \nodeSkipList{i}, \nodeSkipList{j}, \nodeSkipList{k} at any level, we have $d(\nodeSkipList{i}, \nodeSkipList{k})=d(\nodeSkipList{i}, \nodeSkipList{j})+d(\nodeSkipList{j}, \nodeSkipList{k})+1$.
        \item For the \nodeSkipList{i} at any level $t$,  the following holds:
        \begin{equation*}
            d(\nodeNeighbourLeft{i}{t},\nodeNeighbourRight{i}{t})=d(\nodeSkipList{i},\nodeNeighbourRight{i}{t})+d(\nodeNeighbourLeft{i}{t},\nodeSkipList{i})+1
        \end{equation*}
    \end{enumerate} 
    The insertion procedure guarantees the following: If \nodeSkipList{i} has to be inserted at level $l$ then for all levels $1\leq t \leq l$ the following statement holds, 
    \begin{equation*}
        d(\nodeNeighbourLeft{i}{t}, \nodeSkipList{i})=\underbrace{d(\nodeNeighbourLeft{i}{t}, \nodeNeighbourLeft{i}{t-1})}_{\text{Part 1}}+\underbrace{d(\nodeNeighbourLeft{i}{t-1}, \nodeSkipList{i})}_{\text{Part 2}}+1.
    \end{equation*}
    Using the above-mentioned formula one can calculate the value for the $t$-level variables of $\nodeSkipList{i}$ and its neighbours recursively.  Calculation of `Part 1' involves accessing some $(t-1)$-level variable values that were computed even before \nodeSkipList{i} was  inserted, because $(\nodeNeighbourLeft{i}{t-1})\xrightarrow[]{t-1}\textit{left}$ is in actuality \nodeNeighbourLeft{i}{t}. The calculation of `Part 2' requires a recursion. Recall that in our insertion process we store the pointers to the left and right neighbours of \nodeSkipList{i} at every level, so that, once the distance $d(\nodeNeighbourLeft{i}{t},\nodeSkipList{i})$ is computed for any $t\leq l$,  we can use those pointers to update the corresponding left $t$-level variables of the \nodeSkipList{i} and the right $t$-level variable of \nodeNeighbourLeft{i}{t} to the value $d(\nodeNeighbourLeft{i}{t},\nodeSkipList{i})$. For any $t\leq l$, calculating $d(\nodeSkipList{i},\nodeNeighbourRight{i}{t})$ is easy once $d(\nodeNeighbourLeft{i}{t},\nodeSkipList{i})$ is known because
    \begin{equation*}
       d(\nodeSkipList{i},\nodeNeighbourRight{i}{t})=\underbrace{d(\nodeNeighbourLeft{i}{t},\nodeNeighbourRight{i}{t})}_{\text{distance value prior to insertion}}-d(\nodeNeighbourLeft{i}{t},\nodeSkipList{i}). 
    \end{equation*}
    Furthermore, the update can be done in $\log n$ steps because the pointers to the right neighbours of \nodeSkipList{i} were also stored. Therefore, in addition to Ambainis' insertion subroutine, the number of steps required to update all these $t$-level variables of the \nodeSkipList{i} and its neighbouring nodes is $O((\log n)^2)$.\footnote{Note that this recursive procedure can also be implemented iteratively in $O(\log n)$ steps using an additional $O(\log n)$ size array.}
    
    \item \textbf{Deletion.} This process is very simple. Start with the highest level by setting $t=l_{max}$, if the \nodeSkipList{i} is not at level $t$ then reduce the right (left) value of the left (right) possible neighbour of \nodeSkipList{i}. Update $t=t-1$, repeat this process until reaching a $t$ where \nodeSkipList{i} is in fact inserted. All we need to do is to traverse until reaching the left and right neighbours of \nodeSkipList{i} and delete \nodeSkipList{i} from that level and all the subsequent levels. However, before deleting the node, note the sum of $t$-level values of \nodeSkipList{i}. Set the right(left) $t$-level value for the \nodeNeighbourLeft{i}{t}(\nodeNeighbourRight{i}{t}) to the sum after deleting this node from level $t$. Once again update $t=t-1$ and repeat this process until reaching the lowest level $0$.
\end{enumerate}

\paragraph{Indexing procedure and its time complexity.} It is to implement efficient indexing that we modified the skip list data structure used by \cite{Ambainis-ElementDistinctness-2004}. Recall that every node in the skip list now has two variables dedicated to every level $t$ for all $0\leq t\leq (\ceil{\log n}+1)$. Let us denote these variables as $\variableForIndexing{t}{l},\variableForIndexing{t}{r}$, where $t$ indicates the level, $l$ and $r$ indicate the left and right variables, respectively. Given a node $\nodeSkipList{i}$ the variable $\variableForIndexing{t}{l}$ ($\variableForIndexing{t}{r}$) corresponding to this node (which we will refer to as $\nodeSkipList{i}.\variableForIndexing{t}{l/r}$) indicates the number of elements present at the lowest level between \nodeSkipList{i} and its left (right) $t$-level neighbour \nodeNeighbourLeft{i}{t} (\nodeNeighbourRight{i}{t}).

Say given a value $k$, we would like to know the $k\pth$ element of the list. The procedure is identical to finding a particular element, the only difference being that we compare the value of the variables associated with every node ($\nodeSkipList{i}.\variableForIndexing{t}{l/r}$) instead of the node's query value $x_i$. 

The procedure is as follows: 
\begin{enumerate}
    \item Start at level $t=\ceil{\log n}+1$, let the first node that we encounter at this level be some node $\nodeSkipList{i}$.
    \item Check if $k=\nodeSkipList{i}.\variableForIndexing{t}{l}+1$. If yes then return the value of the \nodeSkipList{i}.
    \item If $k\neq \nodeSkipList{i}.\variableForIndexing{t}{l}+1$ then:
    \begin{enumerate}
        \item If $k\leq \nodeSkipList{i}.\variableForIndexing{t}{l}$ then it is guaranteed that the $k\pth$ level element is smaller than the element stored in \nodeSkipList{i}, hence, take a step down and then traverse the new level to the first left neighbour of \nodeSkipList{i} at level $t-1$ that is \nodeNeighbourLeft{i}{t-1}. Then repeat the process from (2) by setting $\nodeSkipList{i}= \nodeNeighbourLeft{i}{t-1}$.
        \item If $k> \nodeSkipList{i}.\variableForIndexing{t}{l}$ then repeat the process from (2) by setting $k=k-\nodeSkipList{i}.\variableForIndexing{l}{t}$ and  $\nodeSkipList{i}=\nodeNeighbourRight{i}{t}$.
    \end{enumerate}
\end{enumerate}

We will soon see that the probability that any indexing operation takes too long is small.

Additionally, to ensure that the quantum walk algorithm has a guaranteed running time, similar to Ambainis' walk algorithm for Element Distinctness, we also modify the insertion, deletion, finding or indexing algorithm to abort if takes more than $c\log^4(n+m)$ steps\footnote{Note that our reduction works even if we modify our algorithm to abort the data structure operations that took more than $n^{o(1)}$ time.}. We observe that this modification has no significant effect on the quantum walk algorithm, details of which are discussed below.

\paragraph{Overview of the errors.}
\begin{enumerate}
    \item As a part of our reduction, we invoke a bounded-error quantum algorithm that solves \sortedThreeSUM{} (presumably) in sublinear amount of time with error probability $\leq 1/3$ as a subroutine. We can cheaply reduce this error to any small $\varepsilon >0$ by running the subroutine  $O(\log (1/\varepsilon))$ times. We need to choose $\varepsilon$ such that $(1-\varepsilon)^{t_2}=1-o(1)$. (Here $t_2$ is the number of times the subroutine is called.) Given that $t_2=(n/r)^{1.5}$, with $r$ polynomially related to $n$, it suffices to choose $\varepsilon=(1/\poly (n))$. This means that the total number of times we will invoke the \sortedThreeSUM{} subroutine to compute the original \threeSUM{} problem is $O(\log n)\cdot t_2$. As we ignore all  poly-logarithmic factors of $n$ in our analysis, we can easily afford to pay for reducing the error probability for the bounded-error \sortedThreeSUM{} subroutine.
    
    \item The data structure we use is probabilistic, therefore occasionally the data structure operations take too long or fail. We will show that, for any $r$-sized set $S'$ that is stored on this modified skip list data structure, the probability that insertion, deletion, finding or indexing takes more than $O(\log^4 n)$ steps is small.
\end{enumerate}

\paragraph{Detailed analysis of errors.} 
Our analysis almost directly follows from the error analysis presented in Section 6 by \cite{Ambainis-ElementDistinctness-2004}.  

We first distinguish between two types of data structure operations (\DSone, \DStwo), the first type (\DSone) constitutes all the data structure operations that take place as the part of the Ambainis' original walk algorithm, the second type (\DStwo) constitutes all the indexing operations that simulate query access to the inputs to the \sortedThreeSUM{} subroutine. The operations of type \DStwo{} are unique to our algorithm. 

Following similar notation as Section 6 of \cite{Ambainis-ElementDistinctness-2004}, we define
\begin{equation}
    \ket{\psi_t}=\sum_{S', y, h_1,..., h_{l_{max}}}\alpha^{t}_{S',y}\ket{\psi_{S',y,h_1,...,h_{l_{max}}}}\ket{y}\ket{h_1,...,h_{l_{max}}}
\end{equation}
as the state of the algorithm after $t$ steps of \emph{perfect} data structure operations (including both types \DSone{} and \DStwo{}) which could take (a lot) more than $\poly \log n$ time but are operations that definitely terminate. Let $\ket{\psi_t}=\ket{\psi^{\textit{good}}_t}+\ket{\psi^{\textit{bad}}_t}$ be the decomposition of state $\ket{\psi_t}$ into \emph{good} and \emph{bad} (un-normalised) states where the good state consists of all terms indexed by $(S', h_1,...h_{l_{max}})$ on which the next (i.e.~the $(t+1)\pth$) operation terminates successfully in $c\log^4(n+m)$ steps (for some fixed constant $c$). Analogously, the bad state consists of all terms indexed by $(S', h_1,...h_{l_{max}})$ on which the $(t+1)\pth$ operation does not terminate in $c\log^4(n+m)$ steps. 

Let $\ket{\psi_{t}'}$ be the state of the quantum algorithm after $t$ (possibly \emph{imperfect}) data structure operations with none of them taking more than $c\log^4(n+m)$ steps. By the term \emph{imperfect} we mean operations that are abruptly terminated in $c\log^4(n+m)$ steps. Combining the results of Lemmas~\ref{lem:traceDistanceBetweenPerfectImperfectStateBounded} and \ref{lem:traceDistanceBetweenPerfectImperfectStateSmall} we show that independent of the value or type (whether \DSone{} or \DStwo{}) of $t$, the trace distance between the two states $\ket{\psi_t}$ and $\ket{\psi_t'}$ is small.

\begin{lem}[Reformulation of Lemma 5 by \cite{Ambainis-ElementDistinctness-2004}]
\label{lem:traceDistanceBetweenPerfectImperfectStateBounded}
For both types of data structure operations (\DSone{},\DStwo{}), we have
\begin{equation}
    \traceNorm{\psi_t}{\psi_t'} \leq \sum_{i=1}^{t-1}||2\psi^{\textit{bad}}_i||.
\end{equation}
\end{lem}
\begin{proof}
The proof of Lemma 5 by \cite{Ambainis-ElementDistinctness-2004} applies to both \DSone{} and \DStwo{} types of data structure operations. 

Let $\ket{\psi_t''}$ be an intermediary state that is obtained with first $(t-1)$ \emph{perfect} data structure operations of any of \DSone{} or \DStwo{} type, but, the last $t\pth$ operation could be \emph{imperfect} as it terminates in $c\log^4(n+m)$ steps. Then for any $t$,
\begin{equation}
\label{eq:IntermediateState1}
    \traceNorm{\psi_t}{\psi_t'} \leq \traceNorm{\psi_t}{\psi_t''}+\traceNorm{\psi_t''}{\psi_t'}.
\end{equation}
The second term on the right hand side of the equation, $\traceNorm{\psi_t''}{\psi_t'}$ is equal to  $\traceNorm{\psi_{t-1}}{\psi_{t-1}'}$ because $\ket{\psi_t''}$ and $\ket{\psi_t'}$ are obtained by applying same (possibly imperfect) unitary transformation to $\ket{\psi_{t-1}}$ and $\ket{\psi_{t-1}'}$, respectively.

The analysis of the first term however is slightly more complicated and is as follows: Let $U_p$ and $U_i$ represent the perfect and imperfect ways of implementing the $t\pth$ data structure operation. Then $\ket{\psi_t}=U_p\ket{\psi_{t-1}}=U_p\ket{\psi^{\textit{good}}_{t-1}}+U_p\ket{\psi^{\textit{bad}}_{t-1}}$ and $\ket{\psi_t''}=U_i\ket{\psi_{t-1}}=U_i\ket{\psi^{\textit{good}}_{t-1}}+U_i\ket{\psi^{\textit{bad}}_{t-1}}$. Clearly, $U_p\ket{\psi^{\textit{good}}_{t-1}}=U_i\ket{\psi^{\textit{good}}_{t-1}}$ because of how $\ket{\psi^{\textit{good}}_{t-1}}$ state is defined. Therefore,
\begin{equation}
\label{eq:IntermediateState2}
    \traceNorm{\psi_t}{\psi_t''} = \traceNorm{U_p \ket{\psi_{t-1}}}{U_i \ket{\psi_{t-1}}} = \traceNorm{U_p \ket{\psi^{\textit{bad}}_{t-1}}}{U_i \ket{\psi^{\textit{bad}}_{t-1}}}\leq 2||\psi^{\textit{bad}}_{t-1}||.
\end{equation}
Combining Equations \ref{eq:IntermediateState1} and \ref{eq:IntermediateState2} we get 
\begin{equation}
    \traceNorm{\psi_t}{\psi_t'} \leq \traceNorm{\psi_{t-1}}{\psi_{t-1}'}+2||\psi^{\textit{bad}}_{t-1}||,
\end{equation}
therefore, inductively proving the statement of Lemma \ref{lem:traceDistanceBetweenPerfectImperfectStateBounded}.
\end{proof}

We will now show that, for any $t$, no matter the type (\DSone{} or \DStwo{}) of data structure operation, the value of $||\psi^{\textit{bad}}_t||$ is small.

\begin{lem}[Modified version of Lemma 6 by \cite{Ambainis-ElementDistinctness-2004} to include \DStwo{} type operations]
\label{lem:traceDistanceBetweenPerfectImperfectStateSmall}
For every $t$, irrespective of the operation being of type \DSone{} or \DStwo{},
\begin{equation}
    ||\psi^{\textit{bad}}_t|| = O(\frac{1}{n^{1.5}}).
\end{equation}
\end{lem}
\begin{proof}
Let again $\ket{\psi_t}$ denote the state of the algorithm after $t$ perfect data structure operations. More formally,
\begin{equation}
\label{eq:StateDescription}
    \ket{\psi_t}=\sum_{S', y, h_1,..., h_{l_{max}}}\alpha^{t}_{S',y}\ket{\psi_{S',y,h_1,...,h_{l_{max}}}}\ket{y}\ket{h_1,...,h_{l_{max}}}.
\end{equation} 
Ambainis' quantum walk algorithm for Element Distinctness that we apply ensures that every basis state $\ket{S',y}$ of same type\footnote{In the context of the constant-sized-subset finding problem where the constant is $k$, the type of a basis state $\ket{S',y}$ is parameterized by $(s,b)$, where $s=|S''\cap \{i_1,...,i_k\}|$ and $b=|\{y\} \cap \{i_1,...,i_k\}|$ for the set $S''\subseteq [n]$ with $|S''|=r$ and $y \in [n] \setminus S''$, with $S''$ denoting the set of indices from the set $S'$. Therefore, there are total of $(2k+1)$ types of basis state. One can also refer to the result in Lemma 1 and 6 by \cite{Ambainis-ElementDistinctness-2004} for the definition of type. } has the same amplitude throughout the algorithm, furthermore, all $h_1,...,h_{l_{max}}$ have equal probabilities. Therefore, we can re-write Equation \ref{eq:StateDescription} as
\begin{equation}
\label{eq:StateDescription1}
    \ket{\psi_t}=\sum_{g}\alpha^{t}_{g} \sum_{Type(S',y)=g, h_1,..., h_{l_{max}}}\ket{\psi_{S',y,h_1,...,h_{l_{max}}}}\ket{y}\ket{h_1,...,h_{l_{max}}},
\end{equation} 
where $g=(\cdot, \cdot)$ denotes a type and $|g|=7$ in the context of the \threeSUM{} problem. 

For the \DSone{} type operations errors can occur in the following two ways:
\begin{enumerate}
    \item The hash table used as mentioned in Equation~\ref{eq:HashTableForUniquenessCondt} can overflow, as more than $\ceil{\log n}$ elements could get hashed to the same bucket. It is shown in the proof of Lemma~6 by \cite{Ambainis-ElementDistinctness-2004}) that, given a type $g=(s,b)$, the fraction of basis states $\ket{S',y}$ of type $g$ for which this operation fails is $o(\frac{1}{n^4})$.
    
    \item The other error stems from the fact that the insertion, deletion, or lookup on the skip list could take more than $c\log^4(n+m)$ steps. We will soon see that, for any fixed $S'$, the probability of that happening is $O(\frac{1}{n^3})$. One can either refer to the proof of Lemma~\ref{lem:AnalysisForDSTwoOps} mentioned below or to Lemma~6 by \cite{Ambainis-ElementDistinctness-2004} for the analysis. 
\end{enumerate}

To summarize: Ambainis in \cite{Ambainis-ElementDistinctness-2004} shows that for any fixed type $g$, the fraction of the states $\ket{S',y,h_1,...,h_{l_{max}}}$ satisfying $\textit{Type}(S',y)=g$ on which the \DSone{} type data structure operations fail is $O(\frac{1}{n^3})$. 

Furthermore, one can re-write the state $\ket{\psi_t}$ from Equation~\ref{eq:StateDescription1} as

\begin{equation}
\label{eq:StateDescription2}
    \sum_{g}\alpha^{t}_{g}( \underbrace{\sum_{\substack{Type(S',y)=g, \\  h_1,..., h_{l_{max}}, \\ \textit{Bad states}}}\ket{\psi_{S',y,h_1,...,h_{l_{max}}}}\ket{y}\ket{h_1,...,h_{l_{max}}}}_{\textit{Small fraction of states}}+\sum_{\substack{Type(S',y)=g, \\ h_1,..., h_{l_{max}}, \\ \textit{Good states}}}\ket{\psi_{S',y,h_1,...,h_{l_{max}}}}\ket{y}\ket{h_1,...,h_{l_{max}}}).
\end{equation} 

Therefore, for any $t$ that is a data structure operation of type \DSone{} we have $||\psi^{\textit{bad}}_t||^2=O(\frac{1}{n^3})$. We will now show that the same result holds for the data structure operation of type \DStwo{} as well.

\paragraph{Detailed error analysis for the \DStwo{} type data structure operations.} We now focus on the errors that stem from the \DStwo{} type data structure operation unique to our algorithm. In the context of a quantum walk on the Johnson graph $J(n,r)$, these operations correspond with checking whether a node is marked. To implement this check, we invoke an assumed \sortedThreeSUM{} quantum subroutine that presumably takes $O(r^{1-\alpha})$ time on an $r$-sized sorted input. To simulate the query access to this subroutine's input we use efficient indexing for the data structure that stores this input in sorted order. The error stems from  the fact that the data structure we use is probabilistic, as our indexing operation will not always terminate in a fixed time. 

Unlike the analysis before, here we don't need to consider different types separately, because the failure of indexing operations is only determined by the set of hash functions $\{h_1,...,h_{l_{max}}\}$ and is same for any fixed set $S'$. Therefore, it is sufficient to estimate the probability of error of implementing any indexing operation on a fixed set $S'$, irrespective of its type, stored on a skip list data structure determined by a set of randomly chosen hash functions $h_1,...,h_{l_{max}}$. This probability turns out to be ${o}(\frac{1}{n^3})$. More formally,

\begin{lem}[Extension of Lemma~6 by \cite{Ambainis-ElementDistinctness-2004} to include analysis for \DStwo{} operations]
\label{lem:AnalysisForDSTwoOps}
Given a fixed $S'$ stored on a skip list determined by randomly chosen hash functions $(h_1,h_2,...,h_{l_{max}})$ from a family of $d$-wise independent hash functions\footnote{Note that the value of $d$ can be set to $c_1\log n$ for any $c_1$ by just increasing the size of the family to $O(N^{\ceil{d/2}})$. Also see Theorem~1 by \cite{Ambainis-ElementDistinctness-2004}.}, with $d=4\log n+1$ , the probability that any indexing operation accesses more than $c\log^2 n$~pointers  (for some constant $c$)  is $o(\frac{1}{n^3})$.
\end{lem}
\begin{proof}
We are given a set $S'$ such that $|S'|=r$ and $S'$ is stored in a skip list determined by randomly chosen $h_1,...,h_{l_{max}}$ from a family of $d$-wise independent hash functions. The first part of the result is to show that for any index $k \in [r]$, the probability that finding the element with $k\pth$ largest query value from the skip list that stores $S'$ takes too long is $O(\frac{1}{n^4})$. This claim directly follows from the calculations in the proof of Lemma~6 by \cite{Ambainis-ElementDistinctness-2004} which we will now summarize. 

Recall that the skip list stores the set $S'$, which contains elements of the form $(i,x_i)$ with elements first ordered by its query values and then ordered by the index value. Our earlier-mentioned indexing algorithm starts with the level $l=l_{max}$ and finds the two neighbouring elements at this level between which the $k\pth$ largest element of $S'$ resides. Lets call these elements of level $l$ as $E_{\textit{left},l}$ and $E_{\textit{right},l}$. Having found the (relative) position of the $k\pth$ largest element at level $l$, it then traverses the level $(l-1)$ pointers starting from element $E_{\textit{left}}$ to find the elements $E_{\textit{left},(l-1)}$ and $E_{\textit{right},(l-1)}$. Following the similar argument as to that of Ambainis we show that the probability of our indexing algorithm taking more than $c\log^4 (n+m)$ steps is small.

The probability that the indexing algorithm accesses more than $d(=4\log n +1)$ elements at level $l_{max}$ is $\frac{1}{n^{O(\log n)}}$. (In fact for the level $l_{max}$, one can prove something stronger: the probability that there are more than $d$ elements at level $l_{max}$ is $\frac{1}{n^{O(\log n)}}$.) The reason is as follows: For any $i \in [n]$, the element $i$ belongs to the level $l_{max}$ if $h_1(i)=h_2(i)=...=h_{l_{max}}(i)=1$. Suppose that the $d$ elements at level $l_{max}$ are $i_1,i_2,....,i_d$, then for all these elements we have $h_{l}(i_1)=h_{l}(i_2)=...=h_{l}(i_d)=1$ for all $l \in \{1,...,l_{max}\}$. For any $l$, the probability of $h_{l}(i_1)=h_{l}(i_2)=...=h_{l}(i_d)=1$ is equal to $\frac{1}{2^d}$ because $h_l$ is a randomly chosen hash function from a family of $d$-wise independent hash functions. Which means that the probability that for all $l$, $h_{l}(i_1)=h_{l}(i_2)=...=h_{l}(i_d)=1$ is $\frac{1}{2^{d\cdot O(\log n)}}=\frac{1}{n^{O(\log n)}}$.

Now we will argue that the probability that our indexing algorithm visits more than $d$ elements at any level $l$ with $1 \leq l < l_{max}$ is $O(\frac{1}{n^4})$. Suppose that the indexing algorithm is currently traversing the $l$ level pointers starting from the element $E_{\textit{left},l+1}$. Additionally, we are guaranteed that the $k\pth$ largest element is less than the element $E_{\textit{right},l+1}$. The probability that there are more than $d$ elements between $E_{\textit{left},l+1}$ and $E_{\textit{right},l+1}$ at level $l$ is $O(\frac{1}{2^d})$. The reason is as follows: Suppose that the $d$ elements are indexed by $i_1,i_2,...,i_d$. If these elements are present at level $l$, but not present at level $l+1$, then $h_{l+1}(i_1)=h_{l+1}(i_2)=...=h_{l+1}(i_d)=0$, which happens with probability $\frac{1}{2^d}$, and also for all $l' \in \{1,...,l\}$ $h_{l'}(i_1)=h_{l'}(i_2)=...=h_{l'}(i_d)=1$, which happens with probability $\frac{1}{2^{ld}}=O(\frac{1}{2^d})=O(\frac{1}{n^4})$. 

As there are a total of $O(\log n)$ levels, using the union bound we can see that the probability of an indexing operation failing at any level is $O(\frac{\log n}{n^4})$.

Also recall that we run the \sortedThreeSUM{} quantum algorithm on an $r$-sized sorted input, which implies that there are $r$ possible indices that can be queried. Therefore by invoking the union bound again we see that the probability of at least one of these $r$ indexing operations failing on a fixed set $S'$ is $o(\frac{1}{n^3})$. 
\end{proof}

To summarise, what we have proved is the following: Having fixed a set $S'$, the probability that a random choice of hash functions fails any indexing operation is $o(\frac{1}{n^3})$. As all $h_1,...h_{l_{max}}$ have equal probabilities, the quantum state in our walk algorithm has an $o(\frac{1}{n^3})$ fraction of \emph{bad} hash functions associated with every $S'$. Which means that the square of the amplitude of the \emph{bad} part of the state of the algorithm after a \DStwo{} operation is $||\psi^{\textit{bad}}_t||^2=O(\frac{1}{n^3})$. Therefore, we see that the bound of Lemma 6 by \cite{Ambainis-ElementDistinctness-2004} holds for \DStwo{} type data-structure operations as well.

Hence we conclude the following: Irrespective of the $t\pth$ data-structure operation being of type \DSone{} or \DStwo{}, we have $||\psi^{\textit{bad}}_t||=O(\frac{1}{n^{1.5}})$. 
\end{proof}

We have already seen that the quantum walk algorithm that we describe uses at most $O(n^{1-O(1)})$ data-structure operations which means, employing the results of Lemma~\ref{lem:traceDistanceBetweenPerfectImperfectStateBounded}, \ref{lem:traceDistanceBetweenPerfectImperfectStateSmall} and~\ref{lem:AnalysisForDSTwoOps}, we get that the distance between the final state of the algorithm in the ideal scenario, i.e.,~the situation where the data structure operations that took too long were not aborted, and the real scenario, where the data structure operations that took longer than $O(\log^4(n+m))$ were aborted, is at most $O(\frac{1}{n^{0.5}})$. Thereby proving that a sublinear algorithm for \sortedThreeSUM{} (if it exists) can be used to construct a sublinear algorithm for \threeSUM{} which succeeds with very high probability. Thus proving  Theorem~\ref{thm:HardnessOfThreeSumProbabilisticDataStructure} which is the main theorem of this section.
\end{proof}

Given the \QthreeSUMconjecture{}, an immediate implication of Theorem~\ref{thm:HardnessOfThreeSumProbabilisticDataStructure} is that there are no bounded-error quantum sublinear time algorithms to solve \sortedThreeSUM{}. Interestingly enough, we can further extend this result to prove quantum hardness for \uniqueSortedSUM{} problem as well. More precisely, we show the following:

\begin{figure}[t]
    \centering
    \includegraphics[scale=0.50]{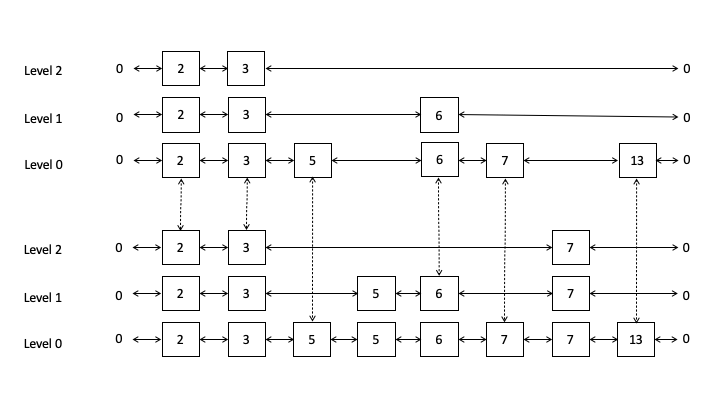}
    \caption{There are two skip lists maintained here, one that contains only distinct queried values, the other contains all queried elements including duplicates. Every element in the first skip list has a pointer to the first occurrence of that element in the second skip list. Note that the nodes in the second skip list also stores the index values, which we do not show in the figure, along with the query values.}
    \label{fig:SkipListUnique}
\end{figure}

\begin{corollary}[to Theorem \ref{thm:HardnessOfThreeSumProbabilisticDataStructure}]
\label{cor:HardnessUniqueSortedSUM}
If there is a bounded-error quantum algorithm that solves \uniqueSortedSUM{}, i.e.~\threeSUM{} with a promise that the input is sorted and distinct, in $\tO(n^{1-\alpha})$ time for some $\alpha >0$ using $g(n)$ qubits of memory, then there exist constants $\beta <1,\delta>0$ such that \threeSUM{} can be solved in $\tO(n^{1-\delta})$ quantum time with high probability using at most $\tO(n^\beta)+g(n^{\beta})$ qubits of memory.
\end{corollary}
\begin{proof} The proof idea and the error analysis are almost identical to the proof of Theorem \ref{thm:HardnessOfThreeSumProbabilisticDataStructure} with only difference being the required data structure. Along with the skip list data structure that we used in the proof of Theorem~\ref{thm:HardnessOfThreeSumProbabilisticDataStructure}, it is possible to maintain an additional skip list that stores only distinct queried values, as described in  Figure~\ref{fig:SkipListUnique}.
\end{proof}

\section{\threeSUM{}-hard Geometry Problems}
\label{sec:AppendixThreeSumHardGeometryProblems}

In this section, we present the quantum reductions from (some variants and structured versions of) \threeSUM{} to many problems in Computational Geometry which are adaptations (or in some case, slight modifications) of the classical reductions by \cite{Overmars-ComputationalGeometry-1995}. 

Most quantum models of computation assume that the input is given as an oracle and can be accessed in superposition, therefore it is possible to have strictly sub-linear quantum time algorithms even for problems that depend on all elements of the input, for example, Grover's search algorithm \cite{GroverSearch-1996}. For the same reasons, it is possible to have quantum reductions which use zero amount of (pre-)computation time as long as query access to the input of the reduced problem is efficiently implementable using the input oracle to the original problem. Therefore, it is possible to quantize these reductions given by \cite{Overmars-ComputationalGeometry-1995} to run in sublinear quantum time even though they take (at~least) linear amount of time classically.

Aaronson et al.~in their paper \cite{Aaronson-ClosestPair-2019} formally define quantum reductions, but, rather than using their more rigorous and general definition we use a quantized version of the notations and terminology used by \cite{Overmars-ComputationalGeometry-1995} because we think they are better suited for our problems and gives more intuition. 

We use \reducedTo{f(n)} to denote a quantum reduction that uses $f(n)$ pre-computation time followed by an efficient on-the-fly reduction for any query. Formally,

\begin{definition}
Given two problems \PRone{} and \PRtwo{} we say \PRone{} reduces to \PRtwo{}, denoted as \PRone{} \reducedTo{f(n)} \PRtwo{}, in $f(n)$ time iff every input of length $n$ to \PRone{} can be solved using a constant number of inputs to \PRtwo{} of length $O(n)$ and $O(f(n))$ pre-computation quantum time, furthermore, query access to the inputs of \PRtwo{} can be implemented efficiently (i.e.~with at most $n^{o(1)}$ overhead) using the query oracle of the input to \PRone{}.
\end{definition}

Consequently, what immediately follows from the definition is the next lemma.

\begin{lem}
Let \PRone{} \reducedTo{f(n)} \PRtwo{}. Let $f(n)$ and $g(n)$ be polynomials. If \PRtwo{} can be solved in $\widetilde{O}(g(n))$ quantum time and $f(n)= \widetilde{O}(g(n))$ and query access to inputs of \PRtwo{} can be efficiently implemented using the query oracle to \PRone{}, then \PRone{} can be solved in $\widetilde{O}(g(n))$ quantum time. Reversely, if $\widetilde{\Omega}(g(n))$ is a quantum lower bound for \PRone{} and $f(n)=O(g(n)^{1-\epsilon})$ for some $\epsilon>0$ and query access to inputs of \PRtwo{} can be efficiently implemented using the query oracle to \PRone{}, then $\widetilde{\Omega}(g(n))$ is
also a quantum lower bound for \PRtwo{}.
\end{lem}

If \PRone{} \reducedTo{f(n)} \PRtwo{} and \PRtwo{} \reducedTo{f(n)} \PRone{} we say that \PRone{} and \PRtwo{} are $f(n)$-equivalent, and we denote that as \PRone{} \equivalentTo{f(n)} \PRtwo{}.

\subsection{Hardness results for some computational geometry problems.}
In this section, we present the reductions that help us to achieve the complexity bounds for some of the computational geometry problems (listed by \cite{Overmars-ComputationalGeometry-1995}) in the quantum setting. One can find the summary of these results in Table~\ref{table:SummaryQuantumSumHardProblems} (in page \pageref{table:SummaryQuantumSumHardProblems}). Most of these reductions are on-the-fly adaptations of the classical ones, hence, we only present the detailed proofs for only a few of these reductions. The proofs for the rest of the reductions are along the same lines. 

\textbf{Problem:} \geomBase{}\\
Given a set of $n$ points with integer coordinates on three horizontal lines $y = 0$, $y = 1$,
and $y = 2$, determine whether there exists a non-horizontal line containing three of the
points.
\begin{thm}
\label{thm:GeomBaseHardness}
\threeSUMthreeList{} \equivalentTo{0} \geomBase{}.
\end{thm}
\begin{proof}
The proof of \threeSUMthreeList{} \reducedTo{0} \geomBase{}: For each element $a \in A$ create a point $(a,0)$, for every $b \in B$ create a point $(b,2)$, and, for every $c \in C$ create a point $(c/2,1)$, which means query access to the instance of \geomBase{} can be directly implemented by using the query oracle of \threeSUMthreeList{} instance. W.l.o.g.~we can assume all the elements of the lists $A,B,C$ are even. (If not then multiply each of these elements with 2.) It is easy to see that three points $(a,0),(b,2)$ and $(c/2,1)$ are collinear iff $a+b=c$, hence, a quantum on-the-fly reduction with $0$ pre-computation time.

The reduction, \geomBase{} \reducedTo{0} \threeSUMthreeList{} is also proved in the similar way, for each point $(a,0)$ create an element $a \in A$, for each point $(b,2)$ create an element $b \in B$, and, for each point $(c,1)$ create an
element $2c \in C$.
\end{proof}

\textbf{Problem:} \threePointsOnLine{} \\
Given a set of points in the plane, is there a line that contains at least three of the points?

\begin{thm}
\uniqueSortedSUM{} \reducedTo{\sqrt{n}} \threePointsOnLine{}.
\end{thm}
\begin{proof}
An input to \uniqueSortedSUM{} is a list $S$ of $n$ unique integers (also sorted but that is not a requirement for this reduction), and the question is whether there exist $a,b,c \in S$ such that $a+b+c=0$. Let $k=2\max(\{|x| \mid x \in S)\}$, this can be computed quantumly in $O(\sqrt{n})$ time. Create a list $S'$ of size $3n$ in the following way: For every $x \in S$, put $x+k, x-3k, x+2k \in S'$. For every element $y \in S'$ create a point $(y,y^3)$. If there exists $a,b,c \in S$ such that $a+b+c=0$ then there will exist a triple $a',b',c' \in S'$ such that $a'+b'+c'=0$ and $a',b',c'$ are all unique.\footnote{The classical reduction from \threeSUM{} to \PointOnThreeLines{} by \cite{Overmars-ComputationalGeometry-1995} is slightly incorrect because of the following counter example: Let $S=\{1,-2,3\}$. There are elements $a, b, c \in S$ such that $a+b+c=0$, set $a=b=1, c=-2$. The classical reduction on such an $S$ is going to create three points $(1,1), (-2,-8), (3,27)$ and will therefore miss out on the \threeSUM{} solution. We rectify this situation by making three (almost) copies of the original list so that solutions to \threeSUM{} where $a=b$ are not missed. The three copies in the intermediate step are deliberately made nonidentical so that the reduction creates two unique points corresponding to $a,b \in S$ even if $a=b$.} Furthermore, with some elementary calculations we can show that $a'+b'+c'=0$ iff $(a',(a')^3),(b',(b')^3),(c',(c')^3)$ are collinear.
\end{proof}

\textbf{Problem: } \PointOnThreeLines{} \\
Given a set of lines in the plane, is there a point that lies on at least three of them?

\begin{thm}
\threePointsOnLine{} \equivalentTo{0} \PointOnThreeLines{}.
\end{thm}
\begin{proof}
Both  these  problems  are  computationally  equivalent  as  the  second  problem  is  the  exact  dual  of the first problem under the Point-Line dualization.
\end{proof}

Lower bounds for the following problems are based on reductions from another promise version of \threeSUM{} (or its variants), namely the \sortedThreeSUM{}. Most of the problems below are reduced from \sortedGeomBase{} instead of \sortedThreeSUM{} because they both are computationally equivalent, which directly follows from the result in Theorem \ref{thm:GeomBaseHardness}.

\begin{corollary}[to Theorem \ref{thm:GeomBaseHardness}]
\sortedThreeSUMThreeList{} \equivalentTo{0} \sortedGeomBase{}.
\end{corollary}

Notice the simple trick that we are employing to adapt these classical reductions in the quantum setting. For all those classical reductions from \threeSUM{} that requires sorting the input, we  directly reduce from \sortedThreeSUM{} instead, because we have shown that \QthreeSUMconjecture{} applies to \sortedThreeSUM{} as well. With this result in spotlight, we present the rest of the reductions.

\textbf{Problem: } \separator{} \\
Given a set $S$ of $n$ possible half-infinite, closed horizontal line segments, is there a non-horizontal separator?

\begin{thm}
\sortedGeomBase{} \reducedTo{0} \separator{}.
\end{thm}
\begin{proof}
The proof directly follows from on-the-fly quantum adaptation of the classical reduction from the \sortedGeomBase{} input to \separator{} problem given by \cite{Overmars-ComputationalGeometry-1995}.
\end{proof}

\textbf{Problem: } \stripsCoverBox{} \\
Given a set of strips in the plane, does their union contain a given axis-parallel rectangle?

\begin{thm}
\sortedGeomBase{} \reducedTo{c} \stripsCoverBox{}.
\end{thm}
\begin{proof}
The proof of this statement also directly follows from the (almost) on-the-fly adaptation of the classical reduction. Each query to the input of \stripsCoverBox{} can be efficiently computed using the query oracle to the input of \sortedGeomBase{}. It is only to compute the boundaries of the rectangle that a constant pre-computation time is required in addition to the on-the-fly reduction.
\end{proof}

\textbf{Problem: } \trianglesCoverTriangle{} \\
Given a set of triangles in the plane, does their union contain another given triangle?

\begin{thm}
\stripsCoverBox{} \reducedTo{0} \trianglesCoverTriangle{}.
\end{thm}
\begin{proof}
The classical reduction given by \cite{Overmars-ComputationalGeometry-1995} is entirely local and can be efficiently made into a quantum on-the-fly reduction, which proves the statement of this theorem.
\end{proof}

\textbf{Problem: }\HoleInUnion{} \\
Given a set of triangles in the plane, does their union contain a hole?

\begin{thm}
\trianglesCoverTriangle{} \reducedTo{0} \HoleInUnion{}.
\end{thm}
\begin{proof}
The classical reduction by \cite{Overmars-ComputationalGeometry-1995} is entirely local and can be efficiently adapted into a quantum on-the-fly reduction.
\end{proof}

The relation between \trianglesCoverTriangle{} and \HoleInUnion{} in the other direction is interesting and is captured in the following statement: 

\begin{thm}
\HoleInUnion{} \reducedTo{n\log^2 n} \trianglesCoverTriangle{}.
\end{thm}
\begin{proof}
This follows by directly using the classical reduction by \cite{Overmars-ComputationalGeometry-1995}.
\end{proof}

Note that, we don't care that the reduction here is \emph{not} strictly sublinear because we are using the following result to give an upper bound for the \HoleInUnion{} problem using the $O(n^{1+o(1)})$ algorithm for \trianglesCoverTriangle{} by \cite{Ambainis-QuantumGeometryProblems-2020} and not to give a lower bound for \trianglesCoverTriangle{} which we already know.

\textbf{Problem: }\TriangleMeasure{} \\
Given a set of triangles in the plane, compute the measure of their union.

\begin{thm}
\trianglesCoverTriangle{} \reducedTo{c} \TriangleMeasure{}.
\end{thm}
\begin{proof}
This follows from the on-the-fly adaptation of the classical reduction by \cite{Overmars-ComputationalGeometry-1995}. A step in the the classical reduction requires that the area of a triangle is computed, this constitutes the constant pre-computation time in our quantum adaptation of the classical reduction.
\end{proof}

\textbf{Problem: }\PointCovering{} \\
Given a set of $n$ halfplanes and a number $k$, determine whether there is a point $p$ that is covered by at least $k$ of the halfplanes.

\begin{thm}
\stripsCoverBox{} \reducedTo{0} \PointCovering{}.
\end{thm}
\begin{proof}
The proof follows directly from a quantum on-the-fly adaptation of the classical reduction by \cite{Overmars-ComputationalGeometry-1995}.
\end{proof}

\textbf{Problem: }\VisibilityBetweenSegments{} \\
Given a set $S$ of $n$ horizontal line segments in the plane and two particular horizontal segments $s_1$ and $s_2$, determine whether there are points on $s_1$ and $s_2$ that can see each
other, that is, such that the open segment between the points does not intersect any segment in $S$.

\begin{thm}
\label{thm:HardnessOfVisibilityOfSegments}
\sortedGeomBase{} \reducedTo{0} \VisibilityBetweenSegments{}.
\end{thm}
\begin{proof}
The proof is a direct quantum on-the-fly adaptation of the classical reduction by \cite{Overmars-ComputationalGeometry-1995}.
\end{proof}

\textbf{Problem: }\VisibilityFromInfinity{} \\
Given a set $S$ of axis-parallel line segments in the plane and one particular horizontal segment $s$, determine whether there is a point on $s$ that can be seen from infinity, that is, whether there exists an infinite ray starting at the point on $s$ that does not intersect any segment.
\begin{thm}
\sortedGeomBase{} \reducedTo{0} \VisibilityFromInfinity{}.
\end{thm}
\begin{proof}
Same as the proof of Theorem~\ref{thm:HardnessOfVisibilityOfSegments}.
\end{proof}

\textbf{Problem: }\VisibleTriangle{} \\
Given a set $S$ of opaque horizontal triangles, another horizontal triangle $t$ and a viewpoint $p$, is there a point on $t$ that can be seen from $p$?

\begin{thm}
\trianglesCoverTriangle{} \reducedTo{0} \VisibleTriangle{}.
\end{thm}
\begin{proof}
The proof follows from the quantum on-the-fly adaptation of the classical reduction by \cite{Overmars-ComputationalGeometry-1995} with the only assumption that the point $p$ of the \VisibleTriangle{} problem is a point at infinity.
\end{proof}

The result in the other direction is only relevant for us to present a quantum upper bound for the \VisibleTriangle{} problem, hence, we directly use the classical reduction to make the following statement.

\begin{thm}[Theorem~7.3 by \cite{Overmars-ComputationalGeometry-1995}] \VisibleTriangle{} \reducedTo{n} \trianglesCoverTriangle{}.
\end{thm}

Due to this result we get a $O(n^{1+o(1)})$ quantum time algorithm for \VisibleTriangle{} using the $O(n^{1+o(1)})$ \trianglesCoverTriangle{} algorithm given \cite{Ambainis-QuantumGeometryProblems-2020}.

\textbf{Problem: }\PlanarMotionPlanning{} \\
Given a set of non-intersecting, non-touching, axis-parallel line segment obstacles in the plane and a line segment robot (a rod or ladder), determine whether the rod can be moved (allowing both translation and rotation) from a given source to a given goal
configuration without colliding with the obstacles.

\begin{thm}
\sortedGeomBase{} \reducedTo{0} \PlanarMotionPlanning{}.
\end{thm}
\begin{proof}
The proof is a direct on-the-fly adaptation of the classical reduction by \cite{Overmars-ComputationalGeometry-1995}.
\end{proof}

\textbf{Problem: }\threeDmotionPlanning \\
Given a set of horizontal (that is, parallel to the xy-plane) non-intersecting, non-touching triangle obstacles in 3D-space, and a vertical line segment as a robot, determine whether the robot can be moved, using translations only, from a source to a goal position without colliding with the obstacles.
\begin{thm}
\trianglesCoverTriangle{} \reducedTo{0} \threeDmotionPlanning.
\end{thm}
\begin{proof}
The proof of this directly follows from the quantum on-the-fly adaptation of the classical reduction reduction by \cite{Overmars-ComputationalGeometry-1995}.
\end{proof}

The following problem, \GeneralCovering{} problem, was introduced by \cite{Ambainis-QuantumGeometryProblems-2020} for which they presented a $O(n^{1+o(1)})$ quantum algorithm. Additionally they showed that many computational geometry problems from \cite{Overmars-ComputationalGeometry-1995} can be solved using the algorithm for \GeneralCovering{}, thereby giving $O(n^{1+o(1)})$ upper bounds for those problems as well. Refer to the summary of these results in Table~\ref{table:SummaryQuantumSumHardProblems}.

\textbf{Problem: }\GeneralCovering{} \\
We are given a set of $n$ strips and angles (angle is defined as an infinite area between two non-parallel lines in the plane). The task is to find a point $X$ that satisfies the following conditions:
\begin{itemize}
    \item the point $X$ is an intersection of two angle or strip boundary lines $l_1$; $l_2$ ($l_1$ and $l_2$ may be boundary lines of two different angles/strips);
    \item the point $X$ does not belong to the interior of any angle or strip;
    \item the point $X$ satisfies a given predicate $P(X)$ that can be computed in $O(1)$ time.
\end{itemize}
\begin{thm}[\cite{Ambainis-QuantumGeometryProblems-2020}]
\stripsCoverBox{} \reducedTo{0} \GeneralCovering{}.
\end{thm}
\begin{proof}
The reduction from \stripsCoverBox{} to \GeneralCovering{} is as follows: Recall that the input to the \stripsCoverBox{} contain $n$ strips and a axis-parallel rectangle. Let the same $n$ strips be input to the \GeneralCovering{} problem. Additionally, as a part of our reduction, we set the predicate 
\begin{equation*}
    P(X)=
    \begin{cases}
        1, & \text{ if $X$ lies in the rectangle,} \\
        0, & \text{ otherwise.}
    \end{cases}  
\end{equation*}
If the \GeneralCovering{} subroutine on this input cannot find such a point $X$ then it implies that the $n$ strips cover the box completely, and, alternatively, if the algorithm for \GeneralCovering{} finds such a point then its clear that the strips don't fully cover the box.
\end{proof}

\section{Other Time Lower Bounds based on \QthreeSUMconjecture{}}
\label{sec:TimeLowerBoundsThreeSum}

In this section, we present a quantum time lower bound of $\Omega(n)$ for the \ConvolutionThreeSum{} problem. We then show that, the classical reduction from \ConvolutionThreeSum{} to \ZeroWtTriangle{} problem can be easily quantized, consequently, proving a $\Omega(n^{1.5})$ time lower bound for the latter. Both of these lower bounds are conditioned on the \QthreeSUMconjecture{}.

The \ZeroWtTriangle{} problem can be solved in $O(n^{1.3})$ queries using the quantum-walk based triangle finding algorithm given by \cite{Magniez-TriangleFindingQuery-2005}. In spite of that, the best known time upper bound for this problem is $O(n^{1.5})$. Our results provide an explanation as to why an $O(n^{1.5-\eps})$ quantum time algorithm for \ZeroWtTriangle{} has not yet been found. 

\subsection{Conditional linear quantum-time lower bound for \ConvolutionThreeSum{}}

Consider the \threeSUMpatrascuVersion{} problem: Given a list $S$ of $n$ elements, is there a $a, b, c \in S$ such that $a+b=c$? We have seen that the \QthreeSUMconjecture{} is equivalent for this version of \threeSUM{}. The \ConvolutionThreeSum{} problem, on the other hand, is defined slightly differently, as follows: Given an array $A[1..n]$, determine if there exists indices $i,j$ such that $i\neq j$ and $A[i]+A[j]=A[i+j]$. There is an obvious $O(n^2)$ classical algorithm and equally obvious $O(n)$ quantum algorithm for \ConvolutionThreeSum{} (and also, less obviously, for \threeSUMpatrascuVersion{}). An interesting classical randomized reduction from \threeSUMpatrascuVersion{} to \ConvolutionThreeSum{} by Pătraşcu \cite{Patrascu-Convoluted3SUM-2010} shows that a subquadratic algorithm for \ConvolutionThreeSum{} implies a subquadratic algorithm for \threeSUMpatrascuVersion{}. We will combine the ideas in that reduction with the quantum-walk-based reduction introduced earlier, to show that the \ConvolutionThreeSum{} problem cannot be solved in $n^{1-\eps}$ quantum time, unless the \QthreeSUMconjecture{} is false.

\paragraph{The reduction by \Patrascu.} In his paper \cite{Patrascu-Convoluted3SUM-2010}, \Patrascu showed a reduction from \threeSUM{} to \ConvolutionThreeSum{}, the intuition of which is as follows. Assume there is an injective hash function $h : S \rightarrow [n]$ which is linear in the sense that $h(a)+h(b)=h(c)$ whenever $a+b=c$. If a such a hash function exists then given an instance of \threeSUMpatrascuVersion{} one can create the \ConvolutionThreeSum{} list by hashing every $a\in S$ to $A[h(a)]$. If there is an $a, b, c \in $ such that $a+b=c$ then by linearity of the hash function, $h(a)+h(b)=h(c)$ which would mean that there exists indices $i=h(a), j=h(b), i+j=h(c)$ such that $A[i]+A[j]=A[i+j]$. Thus, the \threeSUMpatrascuVersion{} triple will be discovered by the \ConvolutionThreeSum{} algorithm. Such a well behaving hash function does not exist, however, it is possible to get something similar.

The reduction uses a family of hash functions introduced by \cite{Dietzfelbinger-AlmostLinearHash-1996} and used by \cite{Patrascu-Convoluted3SUM-2010, Williams-FindingTypesOfTriangles-2009}, defined as follows: Pick a random odd element $z$ on $w$ bits, where $w$ is going to be fixed later. For any input $a \in S$, the hash function multiplies $a$ with $z$ on $w$ bits (i.e., $\mod 2^w$) and then keeps the high order $s$ bits of the result, which can also be visualised in following way: Consider the binary representation of $z a$, pick all the bits between index $w$ to $w-s+1$ (with the lowest significant bit indexed by $1$). Formally, 
\begin{equation}
\label{eq:hashFunction}
    h(a)=(z a \bmod 2^w) \div 2^{w-s},
\end{equation}
where $x \div y$ and $x \bmod y$ denote, respectively, the quotient and remainder of the integer division of $x$ and $y$.
This hash family has the following useful properties.
\begin{enumerate}
    \item \textbf{Almost linear.} For any two numbers $a$ and $b$ either $h(a)+h(b)=h(a+b)(\bmod 2^s)$ or $h(a)+h(b)+1=h(a+b)(\bmod 2^s)$. 
    \item \textbf{Few false positives.} If $a+b=c$, then $h(a)+h(b)+\{0,1\}=h(c)(\bmod 2^s)$ (by which we mean $h(a)+h(b)+b = h(c)(\bmod 2^s)$ for some $b \in \ZO$). Additionally, if $a+b \neq c$, the probability (over the choice of $h$) that $h(a)+h(b)+\{0,1\}=h(c)(\bmod 2^s)$ is $O(1/2^s)$.
    \item \textbf{Good load balancing.} Fix any $n$ elements $z_1, \ldots, z_n$, let us choose a random hash function as above, and let us place $z_i$ into bucket $h(z_i)$. Let $R=2^s$ denote the total number of buckets. Then, over the choice of $h$, any fixed bucket will have $n/R$ elements on average. Also, if we say that a bucket is \textit{bad} if it has more than $3n / R$ elements, then the expected total number of elements that are in bad buckets is $O(R)$.
\end{enumerate}
The classical reduction is given an instance $S$ of $\threeSUMpatrascuVersion{}$, chooses a hash function $h$ at random from the above family, and thinks the elements of $S$ as being placed in ``buckets'', so that $a \in A$ is placed in bucket $h(a)$. 

The reduction then has two parts. The first part deals with the elements of the \textit{bad} buckets, i.e.\ the buckets whose load exceeds $3n/R$. The load-balancing property of this hash function promises that the expected total number of elements in bad buckets is $O(R)$. For every element belonging to a bad bucket, one can in $\widetilde{O}(n)$ classical time decide if it is a part of a solution to the \threeSUMpatrascuVersion{} problem, as follows \cite{Overmars-ComputationalGeometry-1995}: Suppose the list $S$ is sorted (classically, we can afford to sort $S$ at the start). For every element $a$ belonging to a bad bucket compute $S+a$. Using simultaneous traversal of these two ordered sets $S+a$ and $S$ one can in $O(n)$ time find if there is any element common to both. Using this trick for every element belonging to a bad bucket, the entire first part of the reduction then takes $\widetilde{O}(nR)$ classical time.

The second part of the classical reduction creates $O((n/R)^3)$  instances of \ConvolutionThreeSum{} and only deals with elements of the \textit{good} buckets. This part of the reduction is as follows: For every triple $i,j,k \in \{0,...3n/R\}$ we create an instance $A_{i,j,k}$ of \ConvolutionThreeSum{} of size $O(R)$. For each good bucket $t \in [R]$, we map the $i\pth$ element of the $t\pth$ bucket to index $8t+1$, $j\pth$ element to $8t+3$ and $k\pth$ element to $8t+4$. The locations of the array that have no elements mapped to them can have some large value (for e.g., $2 \max(S)+1$) stored in them so that they don't participate in the solution to \threeSUMpatrascuVersion{}. If there was a triple $a,b,c \in S$ such that $a+b=c$ then because of the ``linear"\footnote{The hash function is actually almost linear which will be taken care of in the next paragraph.} hash function we get $t_a+t_b=t_c$ where $t_a=h(a), t_b=h(b)$ and $t_c=h(c)$. This means there exists a triple $i,j,k \in \{0,...3n/R\}$ such that these elements $a,b,c$ get mapped to indices $8t_a+1, 8t_b+3, 8t_c+4$ respectively of the \ConvolutionThreeSum{} array. Hence, the \threeSUMpatrascuVersion{} triple $a,b,c$ is discovered by the \ConvolutionThreeSum{} algorithm. 

Clearly, there will be no false-positives. However, there can be false-negatives: Firstly, because the construction mentioned until now only takes care of all the elements on which the hash function behaved exactly linearly, but as we have stated above, the hash function can actually be off-linear by $1$. The workaround for this is to simply create another set of \ConvolutionThreeSum{} instances where for every bucket $t\in [R]$, instead of mapping the $i\pth$ element of the bucket to index $8t+1$ we map it to $8(t+1)+1$. The second source of false-negatives stems from the fact that \ConvolutionThreeSum{} only checks for $A[i]+A[j] = A[i+j]$, it misses pairs where $h(x) + h(y)\geq R$ (a wrap-around
happens modulo $R$). To fix this, double the array size, including two identical copies. This simulates the wrap-around effect.

\paragraph{Why this reduction doesn't directly hold in the quantum setting.} In the following Subsection~\ref{sec:QuantumWorkaroundForPatrascuReduction} we will see that the first part of the classical reduction that takes $\widetilde{O}(nR)$ can be sped up quantumly to take only $\widetilde{O}(\sqrt{n R})$ time, using the claw-finding algorithm of Buhrman, Dürr, Hoyer, Magniez, Santha, and de Wolf \cite{Buhrman-ClawDetection-2000}. However, the input instance to the claw-finding algorithm needs to be sorted in order for the claw-finding algorithm to run in the required time-bound, and this will correspond to sorting the input $S$ to \threeSUMpatrascuVersion{}.

The second part of the classical reduction also needs  additional structure. The second part of the reduction produces several instances $A_{i,j,k}$ of  \ConvolutionThreeSum{}, and then uses searches for a positive instance among the $A_{i,j,k}$ using an algorithm for \ConvolutionThreeSum{}. This search can be sped-up using Grover search. However, the \ConvolutionThreeSum{} algorithm needs to be able to efficiently read any entry $A_{i,j,k}[\ell]$, and this, in turn, can only be done if we are able to index inside the buckets, i.e., we need to be able to quickly access the $i\pth$ element of the $t\pth$ bucket, for any given $i, t$. The classical algorithm achieves this by simply computing the hash function directly and pre-computing a copy of the input sorted by hash value. This is no longer an option in the quantum setting as doing so requires $\Omega(n)$ time. 

The solution, again, is to use an efficient dynamic data-structure that allows us to both to access a sorted version of $S$, and to index inside the different buckets. We can then use a quantum-walk plus dynamic data-structure, as in the proof of Theorem \ref{thm:HardnessOfOrderedSUM}, in order to adapt the classical reduction to the quantum setting.

\label{sec:QuantumWorkaroundForPatrascuReduction}

\begin{figure}[t]
    \centering
    \includegraphics[scale=0.5]{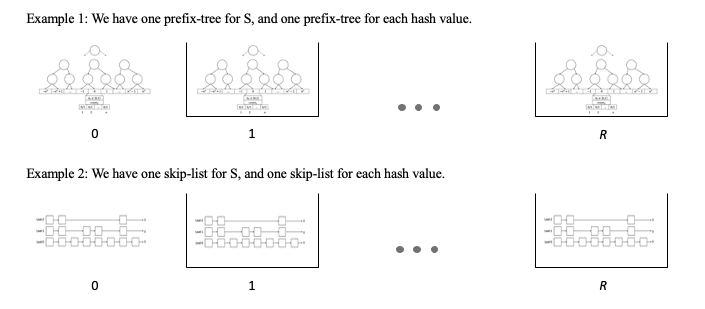}
    \caption{There are two examples of data structures mentioned here. The aim is to store a total of $r$ elements. For the data structure in the first example uses a prefix tree for every bucket, therefore is space inefficient as it requires a total of $\tO(Rn^4)$ memory, however, is deterministic, i.e.~all data structure operations abort in fixed amount of time. The data structure in the second example uses only $\tO(r)$ memory but is probabilistic. The data structures indexed by $0$ in both these examples are used to maintain an additional copy of all the $r$ elements stored in an increasing order of their query values.}
    \label{fig:HashTableEfficientlyAccessible}
\end{figure}

\paragraph{The data structure.} The data-structure works simply as follows: we have maintain $1 + R$ data-structures for dynamic sorting. The first data-structure  is used to maintain $S$ sorted by value, and the remaining $R$ are used to index into the different buckets. With this overall data structure, we can efficiently query the $i\pth$ smallest element of $S$, and we can efficiently query the $i\pth$ smallest element of the $t\pth$ bucket, for any given $i,t$. Suppose that we were given a structured version of \threeSUMpatrascuVersion{}, where the input additionally includes the above data structure. Below, using ideas similar to the classical reduction from \threeSUMpatrascuVersion{} to \ConvolutionThreeSum{}, together with a claw-finding algorithm and Grover search, we will reduce this \emph{structured version} of \threeSUMpatrascuVersion{} to \ConvolutionThreeSum{}. This reduction will give us quantum time lower bound for \ConvolutionThreeSum{} based on the hardness of this structured version of \threeSUMpatrascuVersion{}. 

Once this is done, it will suffice to show the hardness of this structured version of \threeSUMpatrascuVersion{}. This, in turn, can be done using the same quantum-walk-based reduction that was used in the proof of Theorem \ref{thm:HardnessOfOrderedSUM}: we do a quantum walk on the Johnson graph, while dynamically preserving at every step the data-structure illustrated in Figure \ref{fig:HashTableEfficientlyAccessible}. It should be clear that the data-structure can be maintained dynamically, provided we have a data-structure for dynamic sorting within each bucket: when inserting or removing an element, we need only compute its hash value to know into which bucket it should be inserted. We will not go into further details on this part of the reduction, and we will now revisit \Patrascu's reduction from \threeSUMpatrascuVersion{} to \ConvolutionThreeSum{}, in order to show that an analogous reduction can be done, in the quantum setting, from the structured version of \threeSUMpatrascuVersion{} to \ConvolutionThreeSum{}.


\paragraph{Reducing structured \threeSUMpatrascuVersion{} to \ConvolutionThreeSum{}}

Similar to Patrascu's reduction, the first part of our quantum reduction is to deal with the elements of bad buckets which in expectation are at most $O(R)$ in total. We check if there exists any such element that is part of the solution to the \threeSUMpatrascuVersion{} problem. We do Grover search over each element $a$ in a bad bucket, and then we need only find a common element in two sorted lists $S + a$ and $S$ of size $r$ each. Such a common element is called a \textit{claw}, and it is known how to find a  claw in two sorted lists of size $r$ in quantum time $\widetilde{O}(\sqrt{r})$ \cite{Buhrman-ClawDetection-2000}. The total time for this part of the reduction will then be $\widetilde{O}(\sqrt{R r})$.

As explained above, the second part of \Patrascu's reduction from \threeSUMpatrascuVersion{} to \ConvolutionThreeSum{}, creates $O((r/R)^3)$ instances $A_{i,j,k}$ of \ConvolutionThreeSum{}. Each instance is an array of size $|A_{i,j,k}| = O(R)$, and the different instances are indexed by triples $(i,j,k) \in \{0,...,3r/R\}^3$.
The algorithm then checks if there is a solution to at least one of the $A_{i,j,k}$, meaning, two indices $\ell_1,\ell_2 \in [|A_{i,j,k}|]$ such that $A[\ell_1] + A[\ell_2] = A[\ell_1+\ell_2]$. Our quantum reduction will work in the same way, but where we use Grover search to search for a solution among all the triples. 
For this to be possible, we need to provide fast access to each \ConvolutionThreeSum{} instance. Formally, given a triple $i,j,k \in \{0,...,3r/R\}$ and an index $\ell \in [|A_{i,j,k}|]$, we need to be able to quickly return $A_{i,j,k}[\ell]$.
We do the following: Let $\remainder=\ell \mod 8$. If $\remainder \notin \{1,3,4\}$ then return a large value such as $2\max(S)+1$. However if $\remainder \in \{1,3,4\}$ then depending on the value of $\remainder$ return $i\pth$ (if $\remainder=1$) or $j\pth$ (if $\remainder=3$) or $k\pth$ (if $\remainder=4$) element of the $\quotient=\floor{\ell/8}\pth$ bucket. The data-structure is used precisely at this point, in order to efficiently obtain the $i\pth$ (or $j\pth$ or $k\pth$) element of the $\quotient\pth$ bucket.

As mentioned earlier, the buckets (corresponding to the hash function described in Equation~\ref{eq:hashFunction}) along with the support of the data structure (illustrated in Figure~\ref{fig:HashTableEfficientlyAccessible}) for each bucket, allow efficient access to the elements contained in these buckets. For example while using the prefix tree data structure, to access $i\pth$ element of the $t\pth$ bucket one can simply search for the $i\pth$ largest element in the prefix tree of $t\pth$ bucket. Additionally, we also know which bucket is \emph{bad} by looking at the value stored in the root node of the prefix tree at each bucket or in the case of the second example of data structure from Figure~\ref{fig:HashTableEfficientlyAccessible}, one can have a counter variable associated with each bucket that stores the number of elements in that bucket. Access to the information about whether a bucket is bad is useful for implementing the first part of the classical reduction which for the choice of $R=r^\alpha$ for any $0<\alpha<1$ takes strictly sublinear quantum time.

Similarly to the classical reduction, the second part of our quantum reduction checks if a solution exists to at least one of the $O((r/R)^3)$ \ConvolutionThreeSum{} instances. By doing Grover search over the instances, the quantum time complexity of this part is $\tO((r/R)^{1.5}\cdot \ConvolutionThreeSumTime(r'))$ where $\ConvolutionThreeSumTime(r')$ denotes the time taken by a quantum algorithm for \ConvolutionThreeSum{} on a list of $r'$ elements. In this case $r'=O(R)$. 

Therefore, in total, we now have a
\begin{equation}
\label{eq:TimeComplexityViaConvolutionThreeSum}
    \sqrt{rR}+(r/R)^{1.5}\cdot \ConvolutionThreeSumTime(R)
\end{equation}
quantum time algorithm for \threeSUMpatrascuVersion{}, ignoring all the constant and poly-logarithmic factors.

We can now formally state the main result of this section, that is: A sublinear quantum algorithm for \ConvolutionThreeSum{} would imply a sublinear algorithm for \threeSUMpatrascuVersion{}.

\begin{corollary}[to Theorem~\ref{thm:HardnessOfOrderedSUM}]
\label{cor:Convolution3SUMis3SUMhard}
There exists a quantum algorithm for \ConvolutionThreeSum{} running in time $O(n^{1-\delta})$, for some $\delta>0$, if and only if the \QthreeSUMconjecture{} is false.
\end{corollary}

\begin{proof}
We choose the number of hash values to be $R=r^\alpha$ for some $0<\alpha<1$ to be chosen later. Then let $\ConvolutionThreeSumTime(R)=R^{1-\delta}$, for some fixed $\delta>0$, denote the time taken by a bounded-error quantum time algorithm that solves \ConvolutionThreeSum{} on $O(R)$-sized inputs. The expression in Equation~\ref{eq:TimeComplexityViaConvolutionThreeSum} then becomes of order
\begin{equation}
\label{eq:ConvolutionThreeSUMlowerbound}
    r^{\frac{1+\alpha}{2}}+r^{\frac{3}{2}(1-\alpha)} \cdot r^{\alpha(1-\delta)}
\end{equation}
The first additive term in Equation~\ref{eq:ConvolutionThreeSUMlowerbound} is always sublinear, hence can be ignored. Let us analyze the exponent in the second term, i.e., $\frac{3}{2}-\frac{\alpha}{2}-\alpha \delta$. It is easy to see that for every $\delta >0$, there exists an $\alpha$ such that $0< \frac{1}{1+2\delta} <\alpha <1$ and the expression in Equation~\ref{eq:ConvolutionThreeSUMlowerbound} is $r^{1-\Omega(1)}$.

Therefore, a sublinear quantum algorithm for \ConvolutionThreeSum{}, implies a sublinear algorithm for the structured version of \threeSUMpatrascuVersion{}, which according to the result of Theorem~\ref{thm:HardnessOfOrderedSUM} is not possible unless \QthreeSUMconjecture{} is false. 
\end{proof}

\paragraph{Overview and handling of errors.}
\begin{enumerate}
    \item The original quantum walk based query algorithm to solve \threeSUMpatrascuVersion{} has a success probability of $1-o(1)$ \cite{Andrew-SubsetFinding-2005}.
    \item The \threeSUMpatrascuVersion{} subroutine on the $r$-sized subset (stored on the dynamic data structure) could fail, let the failure probability be $p$. Note that, in the actual walk algorithm, this subroutine is repeated $t_2=O((n/r)^{1.5})$ times. Therefore, the probability that all these subroutines are successful is $(1-p)^{t_2}$ which is approximately $(1-t_2p)$ when $|p|<1$ and $|t_2p|\ll 1$. We will now see that, this indeed is the case and the probability of success can be made arbitrarily close to 1.
    
    We will first estimate the value of $p$ and show that it can be made arbitrarily small. The \threeSUMpatrascuVersion{} subroutine on $r$-sized subsets that are stored on the dynamic data structure (from Figure~\ref{fig:HashTableEfficientlyAccessible}) can be reduced to finding a solution in any of the $O(r/R)^3$ instances of \ConvolutionThreeSum{} for which we supposedly have a bounded-error sublinear quantum algorithm. As we know, given a bounded-error algorithm, we can cheaply reduce this error to any $\varepsilon$ by running this subroutine $O(\log(1/\varepsilon))$ many times. With that, the error probability both for the \ConvolutionThreeSum{} subroutine and the Grover subroutine can be made inverse polynomial in $n$ (of any degree $d$) by just repeating these subroutines $O(\log n)$ times. Therefore, one can make the value of $t_2p$ arbitrarily small by just choosing the right value of $d$, with only logarithmic overhead to the run-time of our algorithm.
    
    \item Finally, we mention the errors that stem from the failure\footnote{The term \emph{failure} is used for operations that required more than $c\log^4(n)$ steps and hence were aborted.} of the data structure operations, which could stem from any of the following: (1) The random hash function, as mentioned in Equation~\ref{eq:hashFunction}, or (2) the probabilistic skip-list data structure used to store the elements in each bucket in a sorted order (Example~2 from Figure~\ref{fig:HashTableEfficientlyAccessible}), both of which we discuss in the subsequent paragraphs.
\end{enumerate}

\paragraph{Failure of data structure operations.} Let us first revisit the hash function from Equation~\ref{eq:hashFunction} and its properties. The hash function on any element $a$ is,
\begin{equation}
\label{eq:hashFunctionLaterStage}
    h(a)=(z a \bmod 2^w) \div 2^{w-s},
\end{equation}
where $z$ is a random odd integer of $w$ bits and $s=\Theta(\log w)$. This hash function is probabilistic and has the possible sources of errors:
\begin{enumerate}
    \item This hash function is \emph{always} (almost) linear, i.e.~for any two numbers $a$ and $b$, $h(a)+h(b)+\{0,1\}=h(a+b)(\bmod 2^s)$. Hence, this is not a source of error, in fact similar to the classical case, we only have to create another set of \ConvolutionThreeSum{} instances with only slight modification to the way the first set of instances are created.
    \item However, with probability $O(1/2^s)$, the hash function creates false positive cases, i.e.~$h(a)+h(b)+\{0,1\}=h(c)(\bmod 2^s)$ even when $a+b\neq c$. For our algorithm to be sublinear in time, the value of $R$ has to be equal to $2^s$ and $r^\alpha$, which means $R=n^\zeta$ for some $0<\zeta<1$, making the probability of false positive cases equal to $O(1/n^\zeta)$. We propose the following way to deal with this situation.
    
    The primary goal is to use the \ConvolutionThreeSum{} subroutine as a black box on instances of size $O(R)$ repeatedly and check if any of these $O(r/R)^{3}$ instance has a positive solution, for which we use Grover's search subroutine. The result of the Grover's subroutine gives out the details of which instance of size $O(R)$ has the solution, whose validity can be checked in $O(R)$ additional time. This would worsen the complexity to 
    \begin{equation}
        \label{eq:TimeComplexityViaConvolutionThreeSumLater}
            \sqrt{rR}+(r/R)^{1.5}\cdot \ConvolutionThreeSumTime(R)+R    
    \end{equation}
    as opposed to what we had in in Equation~\ref{eq:TimeComplexityViaConvolutionThreeSum}. However, given that $R=r^\alpha$, the calculations in the proof of Corollary~\ref{cor:Convolution3SUMis3SUMhard} still goes through.
        
    \item Lastly, the hash function has \emph{good load balancing} property which means expected number of elements in the bad buckets is at most $O(R)$ \cite{Dietzfelbinger-AlmostLinearHash-1996}. Using Markov's inequality we can see that the probability of the number of bad elements exceeding $k$ times the expected number of bad elements is upper bounded by $\frac{1}{k}$. Therefore, even for a $k$ as small as $\poly(\log n)$ the probability of error (in the asymptotic sense) is arbitrarily close to $0$.
\end{enumerate}

\paragraph{Errors stemming from the use of probabilistic skip-list data structure.} The later source of errors emerge because of the probabilistic skip-list data structure used to store elements of each bucket in sorted order. The analysis almost directly follows from the ideas presented in Section~\ref{sec:HardnessOfSortedUniqueSpaceEfficient}, which leads to the following observation: Similar to the technique used in proving Theorem~\ref{thm:HardnessOfThreeSumProbabilisticDataStructure}, here also we modify our algorithm to abort any data structure operation that takes more than $c\log^4(n+m)$ steps. We see that the distance between the final state of the algorithm in the ideal situation and in the situation where we abort the \emph{lengthy} data structure operations is small. There are $r$ elements hashed into $R$ buckets using the hash function from Equation~\ref{eq:hashFunctionLaterStage}. Every such bucket is equipped with a skip-list data structure of size $O(r/R)$ to stores all its elements in a sorted order, refer to Figure~\ref{fig:HashTableEfficientlyAccessible}. Probability that a data structure operation takes too long on any bucket can be reduced to $O(R^k/r^k)$ for any constant $k$ with only constant factor worsening of the runtime of the walk algorithm. Given that our prescribed algorithm runs in $O(n^{1-O(1)})$ time, the number of data structure operations is at most sublinear in $n$. Therefore, employing results of Lemma~\ref{lem:traceDistanceBetweenPerfectImperfectStateBounded},\ref{lem:traceDistanceBetweenPerfectImperfectStateSmall} and~\ref{lem:AnalysisForDSTwoOps} we can see that the distance between the final state of our algorithm and the final state of the ideal algorithm is at most $O(n\frac{R^{k/2}}{r^{k/2}})$. With the right choice of $r, R$ and $k$, which indeed is possible\footnote{Recall that $R=r^{\alpha}$ and $r=n^{\beta}$. The following choices of $r,R$ and $k$ make the error arbitrarily small: Let the \ConvolutionThreeSum{} problem have a bounded-error quantum algorithm that runs in $O(N^{1-\delta})$ for a $\delta>0$ on input size $N$. Our algorithm uses this \ConvolutionThreeSum{} subroutine on sets of size $R=r^\alpha$ for some $0<\alpha<1$. Choosing this $\alpha$ from a range $(\frac{1}{1+2\delta},1)$ gives a $r^{1-O(1)}$ algorithm to solve \threeSUMpatrascuVersion{} on $r$ sized subsets, let the exact exponent for the same be $r^{1-\delta'}$ with $\delta'>0$. Using results of Theorem~\ref{thm:HardnessOfOrderedSUM} we can see that for a $\beta \in (\max(\frac{1}{2},\frac{1}{2\delta'+1}),1)$ we could get a sublinear quantum time algorithm for \threeSUMpatrascuVersion{} problem. An additional restriction that needs to imposed here is that, the errors need to be arbitrarily small, which means we would like $n\frac{R^{k/2}}{r^{k/2}}=1/\poly(n)$. For that we need to choose $\beta > \frac{2}{k(1-\alpha)}$. Therefore, it suffices to chose a $\beta \in (\max(\frac{1}{2},\frac{1}{2\delta'+1}, \frac{2}{k(1-\alpha)}),1)$. Note that, $\frac{2}{k(1-\alpha)}$ is not always less than $1$. However, given that $\alpha$ is only dependent on the value of $\delta$ and has to be in the range $(\frac{1}{1+2\delta},1)$, by choosing a $k$ strictly greater than $\frac{2}{1-\alpha}$, we can ensure that there exists a $\beta$ in the range $(\max(\frac{1}{2},\frac{1}{2\delta'+1}, \frac{2}{k(1-\alpha)}),1)$.}, this value $O(n\frac{R^{k/2}}{r^{k/2}})$ becomes $O(\frac{1}{\poly(n)})$.

\subsection{Conditional Quantum $\Omega(n^{1.5})$ Time Bound for \ZeroWtTriangle{} problem}

The quantum reduction from \ConvolutionThreeSum{} to \ZeroWtTriangle{} problem is a straightforward adaptation of the classical local reduction by \cite{Williams-FindingTypesOfTriangles-2009}, which is as follows: Given an input instance of \ConvolutionThreeSum{}, an array $A$ of $n$ elements, the reduction, for every $i \in [\sqrt{n}]$ creates an instance $G_i$ which is a tripartite graph with a weight function associated with the edges of each graph. We will show that there exists a \ZeroWtTriangle{} in any of these $\sqrt{n}$ graphs, if and only if there exists a solution to the \ConvolutionThreeSum{}.

For every $i \in [\sqrt{n}]$, create a complete tripartite graph $G_i$ of three partitioned sets of nodes $L_i, R_i, S_i$ which contain $\sqrt{n}$ nodes each. Let $L_i[t], R_i[t], S_i[t]$ denote the $t\pth$ node of the partition $L_i, R_i, S_i$, respectively. We then set the weights as follows:
\begin{enumerate}
    \item $w(L_i[s],R_i[t])=A[(s-1)\sqrt{n}+t]$,
    \item $w(R_i[t],S_i[q])=A[(i-1)\sqrt{n}+q-t]$,
    \item $w(L_i[s],S_i[q])=-A[(s+i-2)\sqrt{n}+q]$.
\end{enumerate}
Clearly, if there is a triangle in a graph $G_i$ having zero total edge weight, then the value of the weights are solution to the \ConvolutionThreeSum{} problem. The other direction also holds: Suppose there is a solution to the \ConvolutionThreeSum{} at index $i_1, i_2,i_3$ such that $A[i_1]+A[i_2]=A[i_3]$ then there exists a tripartite graph $G_i$ with a \ZeroWtTriangle{} made by the nodes $L_i[s],R_i[t], S_i[i_2'+t]$ where $i-1 = i_2 \div \sqrt{n}$ and $i_2' = i_2 \mod \sqrt{n}$ are the quotient and rest of the integer division of $i_2$ by $\sqrt n$ (which we are assuming is an integer, without loss of generality), and $s - 1 = i_1 \div \sqrt n$, $t = i1 \mod \sqrt n$. It then holds $i_2=(i-1)\sqrt{n}+i_2'$ with $i_2' \in \{0,...,\sqrt{n}-1\}$ and $i_1=(s-1)\sqrt{n}+t$ for $0\leq t <\sqrt{n}$.

As the reduction is completely local, given an index $i\in [\sqrt{n}]$ and any three indices $s,t,q \in [\sqrt{n}]$ we can in constant time query the weights $w(L_i[s],R_i[t]), w(R_i[t],S_i[q]), w(L_i[s],S_i[q])$ associated with nodes $L_i[s], R_i[t],S_i[q]$.

\medskip\noindent
The following now follows:

\begin{thm}
There is no quantum algorithm for the \ZeroWtTriangle{} problem, running in time $O(n^{1.5 - \epsilon})$ for an $\epsilon>0$, unless \QthreeSUMconjecture{} is false.
\end{thm}

\begin{proof}
Let $T(v)=v^{\beta}$, for some $\beta > 0$, denote the time taken quantumly to compute whether a graph $G=(V,E)$ with $|V|=v$ nodes contains a \ZeroWtTriangle{}. 

Using Grover's subroutine over $\sqrt{n}$ indices, one can in $O(n^{1/4}\cdot T(\sqrt{n}))$ quantum time check if there exists an index $i$ such that the graph $G_i$ contains a \ZeroWtTriangle{}. As argued above, this is equivalent to checking for a solution to \ConvolutionThreeSum{} on $n$ elements. Therefore, by Theorem \ref{cor:Convolution3SUMis3SUMhard}, it is required that $\frac{1}{4}+\frac{\beta}{2} \geq 1$, which is to say, $\beta \geq \frac{3}{2}$, unless the \QthreeSUMconjecture{} is false.
\end{proof}

\section{Acknowledgments}
Subhasree Patro is supported by the Robert Bosch Stiftung. Harry Buhrman, Subhasree Patro, and Florian Speelman are additionally supported by NWO Gravitation grants NETWORKS and QSC, and EU grant QuantAlgo. Bruno Loff's research is supported by National Funds through the Portuguese funding agency, FCT - Fundação para a Ciência e a Tecnologia, within project UIDB/50014/2020.

\bibliographystyle{alpha}
\bibliography{3SUM.bib}

\end{document}